%% file: main.tex
\documentclass[10pt,twocolumn,compsoc,journal,a4paper]{style/IEEEtran}
\newcommand{\extendedversion}{final} 

\input{documentProperties}

\usepackage[nocompress]{cite}

\title{ Response-Time-Optimised Service Deployment: \\
\large MILP Formulations of Piece-wise Linear Functions \\
Approximating Bivariate Mixed-integer Functions
}
\date{}

\author{Matthias~Keller and~Holger~Karl%
\IEEEcompsocitemizethanks{%
\IEEEcompsocthanksitem Manuscript received January xx,yy; revised January xx,yy. 
This work was partially supported by the German Research Foundation (DFG) within the Collaborative Research Centre ``On-The-Fly Computing'' (SFB 901).
\IEEEcompsocthanksitem The authors are with the University
of Paderborn, Warburger Str. 100, 33098 Paderborn, Germany.
E-mail: mkeller@upb.de, holger.karl@upb.de
}}

\begin{document}

\maketitle

\begin{abstract}
A current trend in networking and cloud computing is to provide compute resources at widely distributed sites; this is exemplified by developments such as Network Function Virtualisation.
This paves the way for wide-area service deployments with improved service quality: \eg user-perceived response times can be reduced by offering services at nearby sites.
But always assigning users to the nearest site can be a bad decision if this site is already highly utilised.
This paper formalises two related decisions of \emph{allocating} compute resources at different sites and \emph{assigning} users to them with the goal of minimising the response times while the total number of resources to be allocated is limited -- a non-linear capacitated Facility Location Problem with integrated queuing systems.
To efficiently handle its non-linearity, we introduce five linear problem linearisations 
and adapt the currently best heuristic for a similar scenario to our scenario.
All six approaches are compared in experiments for solution quality and solving time.
Surprisingly, our best optimisation formulation outperforms the heuristic in both time and quality.
Additionally, we evaluate the influence of distributions of available compute resources in the network on the response time: The time was halved for some configurations.
The presented formulation techniques for our problem linearisations are applicable to a broader optimisation domain.
\end{abstract}

\IEEEpeerreviewmaketitle

\input{intro}

\input{related}

\input{problem}

\input{eval}

\section{Conclusion}
This paper investigated the problem of allocating resources at
multiple sites in order to minimise the user-perceived 
response time.  Such a distributed deployment  reduces response
time by half or more compared to a single-site deployment
(\Cref{sec:scenario-varia}).  
Five different formulations of optimisation problems were
presented, trading off quality against solving time. One of these
techniques -- thinned curves -- seems particularly attractive as its
solving times are 
vastly superior at only marginally reduced service
response times. This technique, however, is somewhat sensitive to an
improper choice of basepoints; the surface techniques are much more
robust against a small number of basepoints. 

\begin{inextended}
From the considered factors, the topology has -- as expected -- a
strong influence on the optimal solution and its quality. Also,
severely limited compute resources make the problem hard to solve well; the
common wisdom of queuing theory to provide ample spare capacity is
reinforced here in a distributed setting. Moreover, our results
indicate that jointly considering queuing delay and RTT is indeed
crucial when these two times are roughly of the same order of
magnitude -- otherwise, when one of the two times dominates, simpler
optimisation models suffice to obtain good solutions. In fact,  we
found scenarios where single-site deployments outperformed distributed
deployments. 
\end{inextended}

The presented formulation techniques are not limited to the paper's problem.
Beyond this paper, any optimisation problem having a univariate or
bivariate, non-linear cost function can be linearised by the presented
approaches. We extended known approaches to mixed-integer objective
functions which have not been treated in the literature so far. 

Our formulations are not restricted to convex/concave cost functions.
In particular, a newly presented surface linearisation based on
quadrilaterals instead of  commonly used triangles turned out to be a
promising alternative that is sometimes faster than the triangle-based
formulations and has the same solution quality (\Cref{sec:eval:comp:surfaceform}).

We introduced three greedy algorithms \Call{MaxWeight}{},
\Call{Alloc}{}, \Call{Dealloc}{} for allocation problems where $n$
tokens are placed in $m$~buckets so that the costs are minimised. 
If the cost function~$c(j)$ is decreasing in the number of buckets and convex, these algorithms are optimal.

\begin{inextended}
An advanced queuing system approximating a facility, data centre, or rack of compute nodes, (heterogeneous resources) could improve accuracy; but by how much is the question?

The following unproven thought could improve all approaches using the algorithm \Call{Alloc}{}: 
When calling \Call{Alloc}{} without initial allocation the result is equal to but most likely better as providing an initial allocation as proposed in this paper.
Even if the provided initial allocation is optimal for the used problem linearisation, the optimal allocation of \Call{Alloc}{} might performs better.
Similarly, calling \Call{Dealloc}{} can be replaced by calling \Call{Alloc}{} without initial allocation instead.
\end{inextended}

In summary, this paper not only improves service response times by
optimally allocating resources but also presents optimisation
techniques and greedy algorithms applicable beyond this scenario.

\section*{Acknowledgment}
This work was partially supported by the German Research Foundation~(DFG) within the
Collaborative Research Centre ``On-The-Fly Computing''~(SFB~901).

\bibliographystyle{abbrv}
\bibliography{mendexport2}

\vspace*{-10mm}
\vfill
 
\begin{IEEEbiography}[{\includegraphics[width=1in,height=1.25in,clip,keepaspectratio]{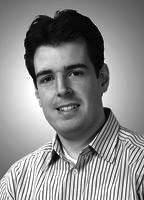}}]{Matthias Keller}
received his diploma degree in computer science from the University of Paderborn, Germany.
He is currently working as a research associate in the Computer Networks group in the University of Paderborn.
He focuses on adaptive resource allocation across 
wide area networks, designed a framework, and created a 
prototype testbed.
Previously, he worked at the Paderborn Centre for Parallel Computing as a research associate and as a software engineer in the industry.
\end{IEEEbiography}

\begin{IEEEbiography}[{\includegraphics[width=1in,height=1.25in,clip,keepaspectratio]{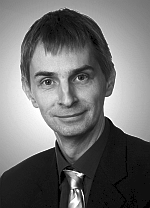}}]{Holger Karl}
received his PhD in 1999 from the Humboldt University Berlin; afterwards he joined the Technical University Berlin. Since 2004, he is Professor for Computer Networks at the University of Paderborn. He is also responsible for the Paderborn Centre for Parallel Computing and has been involved in various European and national research projects. His main research interests are wireless communication and architectures for the Future Internet.
\end{IEEEbiography}

\begin{inextended}
\afterpage{\clearpage}
\include{appendix}

\end{inextended}

\end{document}

%% file: intro.tex
\section{Introduction}
\label{sec:introduction}

\ifextended{\subsection{Opportunities and Problem Statement\label{sec:problem-statement}}}

Providing resources at widely distributed sites is a new trend in networking and cloud computing. It is known under different labels, for example, Carrier Clouds~\cite{Taleb2014,Bagaa2014,Cai2013}, Distributed Cloud Computing~\cite{Endo2011b,Agarwal2010,Pandey2010}, or In-Network Clouds~\cite{Stolyar2013a,Scharf2012,Hao2009}. 
These In-Network Clouds are often geared towards specific network services, \eg, firewalls or load balancers -- a development popularised as Network Function Virtualisation~\cite{Fischer2013}. 
Such deployments have an important advantage: Their compute resources are closer to end users than those of conventional clouds, have hence smaller latency between user and resource, and are therefore suitable for highly interactive services.
Examples for such services are streaming applications~\cite{Wan2010, Barker2010}, user-customised streaming~\cite{Bauer2011, Ishii2011}, or cloud gaming~\cite{Lee2012}.  
For such applications, the crucial quality metric is the user-perceived \emph{response time} to a request. 
Large response times impede usability, increase user frustration~\cite{Claypool2006}, or prevent commercial success.

As detailed in a prior publication~\cite{Keller2014}, these response times comprise three parts: 
A request's \emph{round trip time}~(RTT) in the network, the actual \emph{processing time}~(PT) of a request, and its \emph{queuing delay}~(QD) when it has to wait for free compute resources (\Cref{fig:times_all}).
A simple attempt to provide small response times would be to allocate some resources at many sites so that each user can access the services at a nearby resource\footnotemark.
\footnotetext{\label{fn:resource}Our paper is part of broader research where services are deployed in Virtual Machines~(VM) on geographically distributed sites, each managed by an IaaS-software like OpenStack. 
The work presented by this paper also applies to scenarios  where  a physical host exclusively runs a service.
To emphasis this generality, we  abstractly talk about ``allocating compute resources'' instead of ``launching virtual machines'', without bias for either implementation approach. 
To simplify explanation, we assume allocating homogeneous resources at one side e.g. a slice of \SI{2}{CPUs}, \SI{4}{\giga\byte}~RAM.
Heterogeneous resources can be modelled via simple model transformation (\Cref{sec:model}).
}
This, however, is typically infeasible as each used resource incurs additional costs, e.g.,  consumes energy or has leasing fees.
We hence have to decide
a)~where user requests shall be processed -- the \emph{assignment} decision -- and 
b)~how many resources shall process the requests at each site -- the \emph{allocation} decision.
Both decisions are mutually dependent as exemplified in the next section; the resulting problem is called the \emph{user assignment and server allocation problem}.

\begin{wrapfigure}{tr}{0.31\columnwidth}
  \vspace{-3mm}
  \centering \includegraphics[width=0.30\columnwidth,natwidth=58pt,natheight=70pt]{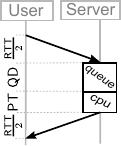}
  \vspace{-1mm}
  \caption{RT = $\quad $ \mbox{RTT + QD + PT}.\label{fig:times_all} }
  \vspace{-3mm}
\end{wrapfigure}
Assignments and allocations influence the queuing delay in two ways:
Allocation modifies the number of allocated resources~$y$ at a site; assignment changes the utilisation of allocated resources at a site (for a fixed~$y$).

Once the number of allocated resources at a site is insufficient to avoid queuing delays for the load assigned to this site, queues will build up and the service quality perceived by the user will suffer. This danger is especially high for compute-intensive services, where even small load can lead to considerable queuing delays. 
Only little research (\Cref{sec:rw}) on server allocation has yet taken these queuing delays into account.

This paper investigates deploying computation-intensive \emph{services} at \emph{compute resources distributed} throughout the net. 
Which distributed compute resource is appropriate depends on a user request's source site (and their request rate).
The service deployment objective is to minimise the response time while using exactly $p$~compute resources.\footnote{We could also limit to \emph{at most} $p$ resources, but since more resources always improve service quality, we use \emph{exactly} $p$ resources.}

When ignoring the queuing delay, assigning user requests to resources allocated  nearby  reduces the round trip time. This problem is hard when exactly $p$ resources are available. 
When ignoring the response time, on the other hand,  the queuing delay is minimised when all requests are assigned to a single site  to which \emph{all} $p$~resources are allocated~\cite{Borst2004}.
Hence, the queuing delay and round trip time are minimised by two conflicting allocation schemes.
To find an optimal solution \emph{for the sum} of queueing delay and round trip time, the two allocation schemes have to be balanced to find assignments and allocations that minimise the response time.

\begin{inextended}
Similarly, the assignment should prefer faster resources, even if their round trip time is a \emph{bit} longer: 
Using faster resources reduces the processing time. 
This time reduction can compensate for the longer round trip time, in total reducing the response time.
\end{inextended}

\begin{notinextended}
Details on how queuing delay affects response time are exemplified in an extended version of this paper~\cite{Keller2015b}. 
It also elaborates further the mutual dependency between assignment and allocation decision.
\end{notinextended}
\begin{inextended}
\subsection{Queuing Delay Effects}
\label{sec:queu-delay-effects}
\begin{table*}[tbp]
\caption{Compares response times considering (\optQP) and ignoring (\optP) QDs for different arrival rates.\label{tbl:example1}}
\centering
\bgroup
\small
\begin{tabular}{S@{\;} S | S@{\;} S@{\;} S S@{\;} S@{\;} S S| S S@{\;} S@{\;} S S@{\;} S@{\;} S} 
  \toprule
  \multicolumn{2}{c|}{arrival rate $\left[ \si{\req\per\second} \right]$} & \multicolumn{3}{c}{\#\,resources} & 
  \multicolumn{3}{c}{time in system $\left[ \si{\milli\second} \right]$} &  { $\varnothing\text{RT}_\text{QP}$} & 
  { $\varnothing\text{RT}_\text{P}$} & \multicolumn{3}{c}{\#\,resources} & 
  \multicolumn{3}{c}{time in system $\left[ \si{\milli\second} \right]$} \\
   {$\lambda_a$} & {$\lambda_b$} & { $f_1$} & { $f_2$} & { $f_3$} & { $f_1$} & { $f_2$} & { $f_3$} & { $\left[ \si{\milli\second} \right]$ } & { $\left[ \si{\milli\second} \right]$} & { $f_1$} & { $f_2$} & { $f_3$} & { $f_1$} & { $f_2$} & { $f_3$} \\ 
  \midrule 
  20 & 10 & 3 & 0 & 2 & 11.1 & 0.0 & 5.6 & 56.7* & 62.2 & 1 & 1 & 3 & 16.7 & 0.0 & 5.6 \\ 
  40 & 30 & 3 & 0 & 2 & 9.7 & 0.0 & 7.6 & 57.3* & 58.3 & 2 & 1 & 2 & 10.7 & 0.0 & 7.6 \\ 
  60 & 50 & 3 & 0 & 2 & 9.5 & 0.0 & 9.2 & 58.7* & 61.3 & 2 & 1 & 2 & 12.1 & 0.0 & 9.2 \\ 
  80 & 70 & 3 & 0 & 2 & 9.9 & 0.0 & 11.8 & 61.6* & 64.3 & 2 & 0 & 3 & 16.0 & 0.0 & 8.3 \\ 
  100 &  90 & 3 & 0 & 2 & 10.7 & 0.0 & 18.0 & 68.8 & 68.8 & 3 & 0 & 2 & 10.7 & 0.0 & 18.0 \\ 
  120 & 110 & 3 & 0 & 2 & 15.8 & 0.0 & 22.0 & 79.3* & 102.5 & 3 & 0 & 2 & 12.6 & 0.0 & 49.9 \\ 
  140 & 130 & 5 & 0 & 0 & 42.1 & 0.0 & 0.0 & 96.5* & 248.9 & 3 & 0 & 2 & 24.2 & 0.0 & 183.3 \\ 
  160 & 150 & 0 & 2 & 3 & 0.0 & 7.9 & 18.4 & 97.6* & 379.1 & 2 & 1 & 2 & 159.7 & 5.3 & 159.7 \\ 
  180 & 170 & 3 & 2 & 0 & 18.9 & 14.5 & 0.0 & 107.2* & 405.9 & 3 & 1 & 1 & 143.1 & 63.1 & 140.0 \\ 
  200 & 190 & 0 & 3 & 2 & 0.0 & 13.4 & 13.0 & 111.3* & 223.3 & 3 & 2 & 0 & 128.5 & 22.0 & 0.0 \\ 
  220 & 210 & 0 & 4 & 1 & 0.0 & 12.6 & 10.6 & 116.3* & 216.4 & 2 & 3 & 0 & 115.1 & 17.7 & 0.0 \\ 
  240 & 230 & 0 & 5 & 0 & 0.0 & 12.4 & 0.0 & 112.4* & 291.3 & 2 & 3 & 0 & 105.3 & 101.0 & 0.0 \\ 
  260 & 250 & 0 & 5 & 0 & 0.0 & 15.6 & 0.0 & 115.6* & 223.3 & 1 & 4 & 0 & 96.1 & 34.1 & 0.0 \\ 
  280 & 270 & 0 & 5 & 0 & 0.0 & 24.3 & 0.0 & 124.3 & 124.3 & 0 & 5 & 0 & 0.0 & 24.3 & 0.0 \\ 
  \bottomrule 
\end{tabular} 
\egroup
\includegraphics[height=6em]{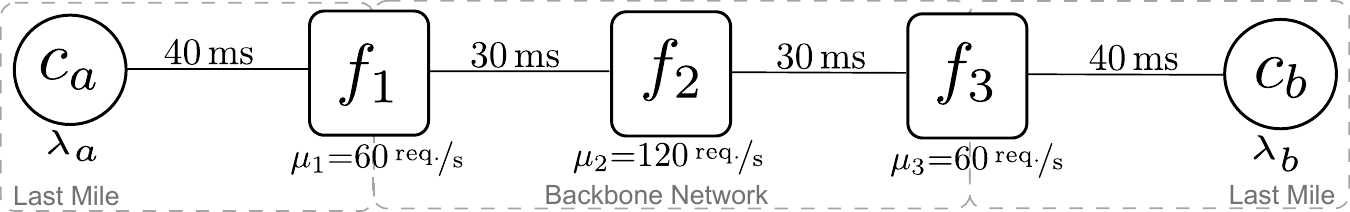}
\end{table*}

Consider graph~$G$ below~\Cref{tbl:example1} as an example:
Two clients $c_a, c_b$ are connected to three facilities $f_{1..3}$ having $10$ resources available, $k_{1..3}{=}10$, such as very large data centres. 
All resources within one facility are homogeneous but they differ among the facilities, having service rates $\mu_{1..3} {=} 60;120;60\,\si{\req\per\second}$ for facility $1;2;3$, respectively.
The assignable requests are limited to $98\%$ of the service rate to avoid very high queuing delays due to its asymptotic behaviour.
This way, the queuing systems are also in steady state.
The number of resources to use are limited to $p{=}5$.
Round trip times between clients and facilities are $l_{a,1..3} {=} l_{b,3..1} {=} 40; 70; 100 \,\si{\req\per\second}$.
Client $c_a$'s arrival rate is slightly larger than $c_b$'s arrival rate; this slight distortion leads to unique solutions.
Each row in \Cref{tbl:example1} corresponds to one level of arrival rates.
In the columns, two solutions are compared which are obtained by either ignoring (\optP) or considering (\optQP) the queuing delays when deciding allocation and assignment.
To compare them directly, for both solutions the following measurements were listed: The average response times $\varnothing \text{RT}$ computed with queuing delays in both cases and the times in queuing systems at each facility.
When comparing the average response times of both solutions in the middle of the table; stars mark the superior solution.
All $\optQP$'s solutions are at least equal to $optP$'s solutions; most are superior.

This superiority has two causes: 
When allocating resources, problem $\optP$ only minimises the round trip time and, thus, assign as many requests as possible nearby. 
Consequently, if all resources are used at the nearest facility, additional resources are allocated somewhere.
In contrast, in our example $\optQP$ allocates additional resources where the queuing delay is reduced the most; cf. $t_{1..3}$ for $\lambda_a {\leq} 60$.
Similarly, the small resource capacity $\mu_1$ forces \optP to assign more assignments to $f_2$ when increasing $\lambda_a {\geq} 200$, $\mu_1 {<} \mu_2$.
In contrast, $\optQP$ assignments are shifted at lower arrival rate to $f_2$ \emph{as the queuing time reduction compensates for the higher round trip time}.
Even with the same allocations $y_{1..3}$ in $\lambda_a{=}120$, $\optQP$ balances the assignments better than $\optP$ 
which results in shorter queueing delays $t_{1..3}$. 

\Cref{tbl:example1} supports another observation: 
The response time differences between \optP and \optQP are small at low or very high system utilisation; \emph{system utilisation} is the fraction of the total arrival rate to the total capacity, $\nicefrac{\sum_f \lambda_f}{\sum_f k_f \mu_f}$.
A high system utilisation enforces one solution where all resources are fully utilised~\eqref{eq:capcomp}.
At low system utilisation, the gain of additionally considering the queuing delay is small, $\lambda_a {\leq} 100$.
However, the tide is turning for medium utilisation, where \optQP's response times can be significantly smaller than \optP's response time, \eg, for $\lambda_a{=} 180$ the $\varnothing\text{RT}_{\optQP}$ is $26.41\%$ of $\varnothing\text{RT}_{\optP}$.%
\begin{IEEEeqnarray}{l}
\underbrace{0.98\,p\,\mu_1 {=} 294}_{\text{cap }f_1} < \underbrace{550}_{\lambda_a{+}\lambda_b} < \underbrace{588 {=} 0.98\,p\,\mu_2}_{\text{cap }f_2}  \label{eq:capcomp} \qquad
\end{IEEEeqnarray}

Assignment and allocation are mutually dependent; making a simple separation into subproblems and solving them sequential is suboptimal.
When assigning demand first, one facility~$f$ could exist whose amount of assigned requests requires allocating more resources as globally allowed (limit $p$); as a simple example imagine $10$~users are assigned to $10$~nearby facilities but only $p{=}8$ resources are allowed to be used. 
They only way to fix this problem is to adjust the assignment to the need of the allocation.
Allocating resources first without assigning requests beforehand is always feasible.
But without knowing and assigning the user requests at the same time, the round trip times between user and facilities with resources are most likely longer than otherwise.
Either way, the solutions are better when assignment and allocation are obtained together.
This in turn significantly increases the complexity of solving the problem to optimality.
\end{inextended}

This paper casts this problem as a queue-extended $p$-median Facility Location Problem~(\Cref{sec:problem}).
To efficiently solve this non-linear optimization problem, five problem linearisations present different modelling approaches, each with different trade-offs between accuracy and search space size~(\Cref{sec:lin}).
Additionally, a well-known heuristic for a similar problem is adjusted to our problem~(\Cref{sec:heuristic}).
\begin{inextended}
Finally, three aspects are evaluated:
First, different basepoint sets for the function linearisations are discussed to reduce the linearisation error and the number of basepoints -- which are two conflicting goals~(\Cref{sec:eval:bps}).
Second, the five linearisations and the heuristic are compared for two metrics, accuracy and solving time~(\Cref{sec:appro-compar}).
Third, an additional evaluation discusses how resource distribution in the network reduces response times~(\Cref{sec:scenario-varia}).
For the last two evaluations, \num{111600} problems are solved.
\end{inextended}
\begin{notinextended}
These five linearisations and the heuristic are compared for two metrics, accuracy and solving time~(\Cref{sec:appro-compar}).
An additional evaluation discusses how resource distribution in the network reduces response times~(\Cref{sec:scenario-varia}).
For the last two evaluations, \num{111600} problems are solved.
In addition, the extended version of this paper~\cite{Keller2015b} discusses linearising a function with high 
accuracy while using few basepoints.
\end{notinextended}

In previous work~\cite{Keller2014}, we had discussed aissgning requests to sites which provided a single compute resource (e.g. a single, fast server). 
The formulated problem was a convex, capacitated queue-extended $p$-median Facility Location Problem.
The present paper extends this work by considering the  compute resources at each site as individual resources  -- now, we additionally decide the resource allocation at each site (e.g., many individual, slower blade servers).
The resulting problem becomes mixed-integer and non-linear,  results in an  M/M/k queuing model instead of an M/M/1 model, limits the allocation to $p$~resources instead of $p$~sites, and considers site capacities in number of available resources with each resource having its own processing capacity. 
Neither problem formulations, linearisation approaches, nor evaluation results presented here have been published before.

%% file: related.tex
\begin{notinextended}

\section{Related work}\label{sec:rw}

Similar assignment and allocation problems with integrated queuing systems have been investigated before.  We focus here on how other work relates to ours.  Details on related papers, a structured comparison of related problems, and a description of related geographical load balancing are given in the extended version of this paper~\cite{Keller2015b}.

\end{notinextended}

\begin{inextended}

\section{Related Work}
\label{sec:rw}
Assignment problems of the form described above have been investigated before. 
We structure their comparison along four dimensions relevant to this work:
Their model complexity, 
simplifications reducing the problem's search space,
optimization goals,
and solution approaches.
Additionally, related systems of geographical load balancing are compared.
Finally, related linearisation techniques are reviewed.

\end{inextended}

\begin{notinextended}

\subsection{Assignment and Allocation with Queues}
Queuing systems have been integrated in FLP problems before~\cite{Wang2002, Berman2006, Aboolian2009, Vidyarthi2009, Pasandideh2010, Moghadas2011, Marianov1996, Marianov2002, Drezner2011} with different objectives.
All of them either have a non-linear objective function or non-linear constraints, but no previous work has utilised a non-linear solver.
We linearised a convex problem in our previous work~\cite{Keller2015}. 
Solving this problem by a linear solver compared to solving the original problem by a convex solver obtains solutions magnitudes faster with only a marginal optimality gap in our experiments.
Another approach uses  cutting planes to improve linearisation accuracy~\cite{Vidyarthi2009} but has an unknown optimality gap.

Other authors solved their related problems via greedy heuristics~\cite{Wang2002,Berman2006,Aboolian2009,Pasandideh2010,Drezner2011} or via full enumeration~\cite{Aboolian2009} for small instances, \eg five facilities. 
The present paper contributes to those papers by describing how to obtain near-optimal solutions for larger input.
We also adapt the genetic heuristic of Aboolian\etalcite{Aboolian2009} to our problem as a comparison case. 
The heuristic was chosen among the references mentioned above because it is well justified, explained, and showed good results.
 
Most mentioned work copes with simpler problems than ours.
Usually, a smaller search space is obtained, \eg, by predefining assignments~\cite{Zhang2012, Aboolian2009, Wang2002, Drezner2011, Berman2006} or by replacing the non-linear queuing delay part by a constant upper bound~\cite{Zhang2012, Wang2002}. 
Both simplifications make load balancing between facilities superfluous in the model, but which could further reduce the average queuing delay in practice.
We do not make such simplifications.

\end{notinextended}

\begin{inextended}

\subsection{Model Complexity}

The simplest model considers \emph{only the round trip time}~(RTT) when assigning users to cloud resources. They equate response time with RTT.
Clearly, this is a simplification of reality, yet minimizing this average RTT is equivalent to the well-known capacitated Facility Location Problem (FLP).
If the problem is further restricted to only use $p$ resources, it becomes a $p$-median FLP, which is NP-hard~\cite{Algorithmic1979}.

A step closer to reality is modelling \emph{also the processing time}~(PT) in addition to the RTT. But as long as PT is constant, this still stays a Facility Location Problem of the type described above. This can be easily seen by extending the original network topology by pseudo-links (at the server or user side) that represent these processing times via their latencies; this is a common rewriting technique for graph-based problems (including Facility Location Problems). 

The real challenge occurs when we also \emph{consider the queuing delay}~(QD).
In this case, the additional time cannot be expressed by rewriting the network topology as the QD depends on the assignment decisions: A higher utilization results in a longer wait, possibly trading off against a shorter RTT.

So far, this more general model has been considered only by few works discussed in the remaining of this section, most use simpler assumptions than ours~(\Cref{sec:problem}) rendering the problem easier to solve.
Vidyarthi\etalcite{Vidyarthi2009} allow the same degrees of freedom as we do.
They approximate, similar to us, the non-linear part of the objective function with a piece-wise linear function. 
However, in contrast to our work, they used a cutting plane technique which iteratively refines the piece-wise function as necessary; it remains unclear how large their linearisation error is.
In contrast, our evaluation shows small linearisation errors; and this is achieved by using a simpler technique.

\subsection{Simplifications}
Other authors investigate slightly different scenarios, so that their problem formulations are similar, yet simpler than ours.

Some authors~\cite{Zhang2012, Wang2002} replace the non-linear QD part with a constant upper bound and, consequently, 
the resulting problems become simpler to solve.
But this also hides QD changes as a result of assignment changes.
For instance, in a situation where load balancing would reduce the QD, this reduction is not visible as the QD part is constant.
Consequently, the resulting solution has further potential for optimization -- we exploit this potential.

In another simplification, the assignments are predefined by a rule. 
Some authors~\cite{Zhang2012, Aboolian2009, Wang2002} always assign requests to the nearest cloud resource.
In such a case, the problem reduces to just finding the best resource location and is easier to solve.
The assignments are then predetermined by the rule.
However, balancing the assignments could further reduce the QD but is not considered.
We do not use any predefined assignment rule, so we have the freedom to change assignments in order to further reduce the response times.

Another group of authors~\cite{Drezner2011, Berman2006} uses a parametrized assignment rule called the gravity rule: Weights determine how users are assigned to cloud resources.
These configurable weights are used to continuously solve the same problem with new weights reflecting the resource utilizations of the previous solution.
This approach does not guarantee to converge, so the authors propose a heuristic that attenuates the changes in each iteration, enforcing convergence with an unknown linearisation error.
In contrast, we solve the problem in one step by using all information to find the global optimum.

Liu\etalcite{Liu2015}, Lin\etalcite{Lin2012}, and Goudarzi\etalcite{Goudarzi2013} present a similar Facility Location Problem with convex costs such as queuing delays or resource's energy costs. 
In contrast to our work, they relax the integer allocation decision variable simplifying the problem to the cost of a less accurate solution when rounding up the obtained continuous allocations.
Our goal, in contrast, is to prevent unexpected expenses by introducing an upper bound to the number of used resources. 
Continuously relaxing our problem can cause any location to be allocated a bit and, consequently, any site is used and paid.
While the papers \cite{Liu2015,Lin2012} only consider queuing delays as a cost function, this paper discusses a holistic queuing system integration and additionally considers splitting and joining (assigning) of the arrival process.

\subsection{Optimization Goal}

Existing literature uses queuing delays in FLPs with three optimization goals: classic FLP, min/max FLP, and coverage FLP.

Classic FLPs are problems that minimize the average response time, like our problem (\Cref{sec:model}) or others~\cite{Wang2002,Berman2006,Vidyarthi2009,Pasandideh2010,Drezner2011,Zhang2012}, allowing RT variations for individual users. 

Aboolian\etalcite{Aboolian2009}'s min/max problem minimizes the maximum response time.
Intuitively, such problems improve especially the users' RT with high RTTs to cloud resources. 
However, if only one such user exists with resources being far away, assigning this user will negatively affect the assignments of other users:
Their assignments are now less-restrictedly constrained by a relaxed upper bound and are likely worse than without the first user.
In contrast, classical FLPs do not suffer this way from a worse case user.

Another type of problem is coverage problems; the user assignment's response times is upper bounded~\cite{Marianov1996,Marianov2002,Moghadas2011}. 
Structurally, a coverage problem is a special, simpler case of a min/max problem; 
the first has a predefined bound, which is additionally minimized in the second. 
Intuitively, such problems can be applied in scenarios where service guarantees for a certain maximal response time will be provided and paid.
In contrast, classical FLPs allow minimizing the average response time below the lowest possible response time bound.

\subsection{Solution Approaches}
A couple of heuristics were proposed solving related problems which are variants of the NP-hard capacitated FLP~\cite{Hakimi1979}. 
No work so far used solvers to obtain solutions (for non-relaxed problems) and full enumerations are known for small instances limited to open five facilities~\cite{Aboolian2009}.
A greedy \emph{dropping} heuristic successively removes from the set of candidates that resource which increases the response time by the smallest amount~\cite{Wang2002}. 
Greedy \emph{adding} heuristics successively add resources, which decreases the response time by the largest amount~\cite{Berman2006,Aboolian2009,Drezner2011}.
Another heuristic probabilistically selects set changes of used resources~\cite{Drezner2011} or performs a breath-first-search through ``neighbouring solutions'' where two solutions are neighbours if their sets of used resources differ in one element~\cite{Aboolian2009,Drezner2011}.
Such heuristics can be stocked in local optima and to mitigate this drawback meta-heuristics are used as a superstructure~\cite{Wang2002,Berman2006,Aboolian2009,Pasandideh2010,Drezner2011}.
These meta-heuristics typically refine previously generated initial solutions, which are obtained randomly or by combining existing solutions. 
The hope is that among the found local optima, one solution is very close to the global optimum -- but without any guarantee.
In contrast, we obtain global optima.
This is an important step for heuristic development as only this enables a clear judgement of heuristics' accuracy; their solution's gap to the global optimum.

Others~\cite{Wang2002,Vidyarthi2009} may achieve near optimal solutions by using optimization techniques like branch-and-bound and cutting planes but their solutions have unknown optimality gabs.
In summary, either optima for small input or solutions with unknown optimality gap are obtained.
This motivated our work on finding near-optimal solutions with a numerically very small optimality gap.

Liu\etalcite{Liu2015} and Wendell\etalcite{Wendell2010} present distributed algorithms for their global Geographical Load Balancing problem by decomposing it into separate subproblems solved by all clients. 
These subproblems converge to the optimal solution only if they are executed in several synchronized rounds in which assignment and utilization information are exchanged among all clients.
Both papers state that this distributed algorithms would obtain optimal solution faster than gathering everything to a centralised solver.
However, we believe that each round a communication delay is introduced when sending update information among all clients; they had ignored this delay in their evaluations. 
The resulting total delay over all rounds is likely larger than communicating with a centralised coordinator. 
In addition, our $p$-median Facility Location Problem has a global constraint on the maximal used resources preventing it to be easily separated into subproblems.

We observed that problem instances were solved only exemplary so far~\cite{Liu2015,Wendell2010,Lin2012,Marianov2002,Berman2006,Aboolian2009,Pasandideh2010,Pasandideh2010,Drezner2011}.
Consequently, the average performance of these solution approaches is hard to predict.
We go beyond this by undertaking a statistical performance evaluation.
We randomly vary our input data and verify the statistical relevance of our findings.

\subsection{Geographical Load Balancing}    

A system for Geographical Load Balancing (GLB) comprises two parts: The decision part selects appropriate server, sites, or Virtual Machines for requests of a certain origin -- the previous sections considers them.
This section focuses on the realisation part, which gathers monitoring information and implements selections.
Different middlewares had been proposed~\cite{Wendell2010,Freeman2006,Wong2006,Wong2005} which are shared between applications.
In this way, each application benefits from instances of the other application running at diverse sites by sharing monitoring information such as latency to servers or to customers.
They realise request assignments, \eg, to close-by or low utilised server by either configuring the Domain Name System (DNS) or are explicitly queried ahead a request send.
Slightly different, Cardellini\etalcite{Cardellini2000} propose redirecting requests to different sites to balance the load. Policies ranges from redirecting all, only largest, or only group requests to selecting sites based on round-robin, site utilization, or connection properties.
Our paper focuses on solving the problem and investigates whether the more complex problem with queuing systems is worth the additional efforts and our results can be applied to improve geographical load balancing systems.

\end{inextended}

\subsection{Linearisation} \label{sec:rw-linear}
Function linearisation is a technique to approximate a non-linear function over a finite interval by line segments; \Cref{sec:pwl} provides technical details.
Applying this technique to objective function or constraints, any non-linear optimisation problem can be approximated by a linear problem~\cite{MINLP2012,DAmbrosio2010}; such a transformation is also called problem linearisation.

We consider solving time and solution quality simultaneously, and, hence, need to choose a suitable number of segments. 
Imamoto\etal's algorithm~\cite{Imamoto2008,Imamoto2008a} (improved by us~\cite{Keller2015}) obtains segments' start and end basepoints yielding a high linearisation accuracy for convex, univariate functions. 
We use this improved version   to obtain basepoints.

Rebennack\etalcite{Rebennack2014} linearised multi-variate functions by first decomposing them into separate independent functions and then recombining the linearisation.
This approach is limited to separable functions. 
But our function of interest is Erlang-C-based~(\Cref{sec:model} Eq.~\ref{eq:erlc}), which is not separable.

Vidyarthi\etalcite{Vidyarthi2009} refines the piece-wise linear function iteratively by adding basepoints while solving.
In contrast, we first compute a tight function linearisation which also involves modifying existing basepoints.
Then, these basepoints are integrated into the problem to solve. 
This two-step approach is much simpler to implement and to solve than changing the search space dynamically during solving.

%% file: problem.tex
\section{Problem} \label{sec:problem}

\begin{inextended}
\subsection{Scenario}\label{sec:scenario}
A service provider deploys a service whose \emph{response time} is a critical metric for service quality, \eg, an interactive service.
Geographically distributed \emph{users} \emph{request} the service.
In order to reduce the request's round trip time, geographically distributed \emph{resources} are used, \eg, services are deployed in Virtual Machines~(VMs) utilising compute resources of data centres at different sites.
Resource performance (\eg host performance executing the VM\footnotemark) and resource utilisation (\eg VM's CPU utilisation) determine how fast an reply is computed. 
\footnotetext{In order to rely on performance indicators in a virtualised environment, the service need to use the host hardware, \eg one CPU, exclusively. Otherwise, multiple VMs share the same hardware and the used performance indicator no more describes the real performance to expect. 
Then the theoretical response time computed by the optimisation largely deviates from the observable response time.}
The sum of this time and the round trip time is the response time of a request -- the user's wait time after sending a request until receiving its reply -- shorter user wait means better service quality.
\end{inextended}

\subsection{Model}\label{sec:model}

The scenario~(\Cref{sec:introduction}) is cast as a capacitated $p$-median Facility Location Problem~\cite{Drezner2004}.
A bipartite graph~$G$ has two types of nodes: \emph{Clients} ($c {\in }C$) and \emph{facilities} ($f {\in} F$). 
Clients correspond to sites where user requests enter the network.
Facilities represent candidate sites with available compute resources\cref{fn:resource} on which a service can be executed, \eg hosts of data centres.
\Cref{fig:problemvis} (p.~\pageref{fig:problemvis}) shows such a graph.
The round trip time~$l_{cf}$ is the time to send data from $c$ to $f$ and back.

The geographically\footnotemark{} distributed demand is modelled by the request arrival rate~$\lambda_c$ for each client~$c$.
\footnotetext{More precisely, the request arrival and service sites are topologically distributed; the round trip time of a path between two sites only roughly matches its geographically distance. We use ``geographical'' as an intuitive shorthand.}
Each facility~$f$ has $k_f$~resources available\footnotemark{} and each resource can process requests at service rate $\mu_f$ -- the resource capacity.

\footnotetext{For example, assuming one resource corresponds to a fixed-size Virtual Machines occupying two CPU cores, a data centre with 30 physical hosts with 4 cores has 60 resources available.}
A facility consists only of homogeneous resources. 
Heterogeneous resources at one facility~$f$ can easily be approximated by partitioning these resources into homogeneous, disjunct resource sets, each represented by a dedicated facility $f_1$, $f_2$, ...; $\forall c{:}\; l_{cf} {=} l_{cf_1} {=} l_{cf_2}$ \ldots; this introduces a small inaccuracy by modelling multiple queues where one queue would be the exact model.
\Cref{tbl:var} lists all variables.

The expected time for computing an answer is obtained by utilising a queuing system at each facility, with the usual  assumptions: 
The service times are exponentially distributed and independent. 
The request arrivals at each client~$c$ are described by a Poisson process; client requests can be assigned to different facilities.
At one facility, requests arrive from different clients and the resulting arrival process is also a Poisson process, because splitting and joining a Poisson process results in a Poisson process.
Therefore, we have an M/M/y-queuing model at each facility.

Having this model at hand, the probability that an arriving request gets queued is computed by the \mbox{Erlang-C}~formula~\funcEC given in, e.g., Eq.~\eqref{eq:erlc} of \cite{Bolch2005}. 
Each facility has to be in steady state with resource utilisation $\rho{ =} \nicefrac{\lambda}{y\mu} {=} \nicefrac{a}{y} {<} 1$.
Derived from \funcEC, the expected \emph{number of requests in the system}~(NiS) is $\funcN(a,\,y)$~\eqref{eq:nis}.
Similarly, the expected \emph{time a request spends in the system}~(TiS) is $\funcT(\lambda,\,y)$~\eqref{eq:tis}.
\begin{IEEEeqnarray}{lll}
\funcEC(a,\,y) &= \frac{\frac{y{a}^{y}}{y!\,\left( y-a\right) }}
    {\frac{y{a}^{y}}{y!\,\left( y-a\right) }+ \sum_{i=0}^{y-1}\frac{{a}^{i}}{i!}} & 
    \text{(Erlang-C)} \label{eq:erlc} \qquad \\
\funcN(a,\,y)  &=   \frac{a}{(y-a)}\,\funcEC(a,\,y) + a \quad  & 
    \quad \text{(NiS)} \label{eq:nis}\\
\funcTT[_{\mu}] \textbf{}(\lambda,\,y) &= \frac{1}{\lambda} \funcN(\nicefrac{\lambda}{\mu},\,y) {=} 
{\small 
\frac{\funcEC\left( \nicefrac{\lambda}{\mu} ,\, y \right) }{y \mu-\lambda} {+} \frac{1}{\mu}  
}
& \quad \text{(TiS)} \label{eq:tis} 
\end{IEEEeqnarray}
\subsection{Formulation}\label{sec:formul}
Two decisions have to be made: (1) Distributing a client $c$'s request rate $\lambda_c$ to one or facilities $f$, each of which gets a request rate~$x_{cf}$ and (2) allocating $y_f$ resources at $f$.  Both influence the queuing delay non-linearly.  \ifextended{ -- a well-known fact later recapitulated in \Cref{theo:convN2}}{} The mutual dependency between both decisions causes the complexity of our problem.  The formulation \optQP~(\Cref{algo:optqp}) extends the $p$-median Facility Location Problem \optP~\cite{Keller2015} by having additional costs, the queuing delays, at each facility.  Constraint~\eqref{eq:qpk:demand} ensures that each client's demand is assigned.  The assignments to each facility $f$, $\sum_c x_{cf}$, must not exceed $f$'s service rate $y_f\mu_f$ \eqref{eq:qpk:cap}, which is the service rate per resource~$\mu_f$ times the number of allocated resources~$y_f$; this constraint also ensures that $f$'s queue is always in steady state.  The allocated resources~$y_f$ do not exceed the local \eqref{eq:qpk:count} and global \eqref{eq:qpk:limit} limit.
\begin{table}[tb]
\vspace{-2mm}
\caption{Model variables\label{tbl:var}}  %
\bgroup
\renewcommand{\arraystretch}{1.3}
\small 
\begin{tabular}{p{18mm}>{\footnotesize}p{62mm}}
\toprule 
\multicolumn{2}{l}{Input:} \\
$G=(V,E,l,$ $ \lambda, \mu, k)$ & Bipartite Graph with $V=C \cup F$, {$ C \cap F = \varnothing$}  with client nodes $c \in C$ and facility nodes $f \in F$ \\
$\lcf \in \SetR$ & Round trip time between $c$ and $f$  \\ 
$\mu_f \in \SetR[>0]$  & Request service rate as resource capacity at $f$ \\ 
$k_f \in \SetN[>0]$ & Number of servers available at $f$  \\
$\lambda_c \in \SetR$   & Request arrival rate as demand at $c$ \\ 
$\alpha_{i};\, \beta_{i}$ & $i$-th basepoint $g(\alpha_i) {=} \beta_i$ of PWL function~$\tilde g$ \\
$p = \sum_f y_f$ & Limit on maximal resources to allocation \\
\midrule
\multicolumn{2}{l}{Decision variables:} \\
$\xcf \in \SetR$  & Assignment, request rate\\ 
$y_f \in \SetN[\geq 0]$ & Number of allocated resources (active servers) at $f$ \\
$\dot y_{f};\, \dot y_{fj} {\in} \{0,1\}$ & Indicator: $f$ open; with active PWL func. $\funcTa[\mu_f,j]$\\
$z_{fi};\, z_{fji} {\in} [0,1] $ & Weight of $i$-th basepoint of PWL func. $\funcTa[\mu_f,j]$\\
$h^u_{fji};\, h^l_{fji} {\in}$ $\{0,1\}$ & Indicator: $(j,i)$-th triangle at $f$ \\
\bottomrule
\multicolumn{2}{>{\footnotesize}p{82mm}}{
    Notation shorthand: $k$ or $k_*$ refers to the tuple/vector~$(k_1,..,k_{|F|})$; 
    and $x_{c*}$ refers to matrix slice~$(x_{c1},..,x_{c|F|})$.  }\\
\end{tabular} 
\egroup
\end{table}
\begin{optproblem}[tbp]
\caption{$\optQP(G,\, p ,\, \funcT[*])$, the reference\label{algo:optqp}}
\ifextended{}{\small}
\begin{IEEEeqnarray}{llls}
\min_{x,y} \, & \IEEEeqnarraymulticol{3}{l}{%
    \underbrace{\frac {\sum_{cf} x_{cf} l_{cf}} {\sum_c \lambda_c}}_{\text{avg. RTT}} +    
    \underbrace{\frac{\sum_{f} ( \sum_c x_{cf} )\, \funcTT[_{\mu_f}] (\sum_c x_{cf}, y_{f} )}{\sum_c \lambda_c}}_{\text{avg. time in system}}   
    }\quad\label{eq:qpk:obj} \\
\text{s.t.} \,  
&   \sum_f x_{cf} = \lambda_{c} & \forall c \quad
    & (demand) \label{eq:qpk:demand} \\
&	\sum_c x_{cf}  < \mu_f y_f \quad & \forall f 
    & (capacity) \label{eq:qpk:cap} \\
&	y_f \leq k_f & \forall f 
    & (count) \label{eq:qpk:count}\\
&	\sum_{f} y_f = p  &   
    &(limit) \label{eq:qpk:limit} 
\end{IEEEeqnarray}
\vspace{-2mm}
\end{optproblem}
\section{Linearisation}\label{sec:lin}
The problem~\optQP is non-linear and could be solved by a non-linear solver.
In previous work~\cite{Keller2015}, a simpler, yet also non-linear version of this problem was successfully linearised and was solved by a linear solver fast and without substantial quality loss.
Because of these encouraging result, we also follow the linearisation approach in this paper.
\optQP is more complex than our problem from the previous paper:
While both problems decide the request assignments, the problem~\optQP additionally decides the number of allocated resources at each site. 
This turns \funcT from the univariate function~$\funcT(\lambda)$ into the bivariate function~$\funcT(\lambda, y)$.
In addition, the second parameter is integer, making it difficult to apply standard linearisation formulations.

The remaining section describes the linearisation technique in general (\Cref{sec:pwl}) and reformulates \optQP by linearising either several curves separately (\Cref{sec:lin:curve}, \Cref{sec:lin:curvethin}) or together as a surface (\Cref{sec:lin:surfacetriangle}, \Cref{sec:lin:surfacequad}, \Cref{sec:post-processing-surfaces-results}).
\begin{notinextended} 
The extended version of this paper~\cite{Keller2015b} additionally discusses problem simplifications only applicable to convex cost functions.
Additionally, it contains technical details on efficient computation techniques for Erlang-C-type functions.
\end{notinextended}

\begin{inextended}
Then, \Cref{sec:expl-struc-knowl}  further discusses potential problem simplifications.
Finally, \Cref{sec:lin-erlangc} presents technical details on how Erlang-C-base functions are efficiently implemented.
\end{inextended}

\subsection{Piece-wise linear functions}\label{sec:pwl}

\begin{figure}[tbp]
\centering
\begin{minipage}{0.48\columnwidth}
\centering
\includegraphics[width=2cm]{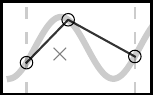}
\caption{1D-PWL example.\label{fig:intro1D}}
\end{minipage}
\hfill
\begin{minipage}{0.48\columnwidth}
\centering
\includegraphics[width=2cm,clip,trim=34pt 22pt 39pt 19pt]{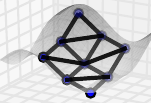} 
\caption{2D-PWL example.\label{fig:intro2d}}
\end{minipage}

\end{figure}

Any non-linear, univariate function~$g$ can be approximated over a finite interval by multiple line segments.
A function composed of such segments is called a piece-wise linear~(PWL) function~\cite{MINLP2012} (strictly speaking, a piece-wise affine function).
We refer to the linearisation of $g$ as $\tilde g$.
As an example, \Cref{fig:intro1D} shows $g(x){=}\sin(x)$ and two segments linearising $g(x)$.
For $m$ basepoints with coordinates $(\alpha_i,\, \beta_i{=}g(\alpha_i))$, $0 {\leq} i {<} m$, a continuous PWL function can be defined~\eqref{eq:pwl1D} by linearly interpolating between two adjacent basepoints (black circles in \Cref{fig:intro1D}).
The interval is implicitly defined by the outer basepoints, $[\alpha_0, \alpha_{m{-}1}]$.

\begin{IEEEeqnarray}{l}
\tilde g(x) = \begin{cases}
(\alpha_1{-}\alpha_0) (x{-}\alpha_0) + \beta_0 & \text{ if } \alpha_0 {\leq} x {<} \alpha_1 \\ 
(\alpha_2{-}\alpha_1) (x{-}\alpha_1) + \beta_1 & \text{ if } \alpha_1 {\leq} x {<} \alpha_2 \\ 
\hspace{3em} \dots \\
\end{cases} \label{eq:pwl1D} 
\end{IEEEeqnarray}

The linearisation accuracy $\epsilon$ is measured by the maximal difference~$\epsilon$ between $g$ and $\tilde g$, $\epsilon {=} \max_x \vert g(x) - \tilde g(x) \vert$. 
Because the PWL function definition in \eqref{eq:pwl1D} is not directly understood by an MILP~solver, it has to be transformed.
We use the convex-combination method~\cite{Dantzig1960} to model continuous, univariate PWL functions: The cases of \eqref{eq:pwl1D} are substituted by a convex combination of the basepoint coordinates~\eqref{eq:pwlconvex} with weights~$z_i$.
\bgroup
\small
\begin{IEEEeqnarray}{ll}
\tilde y {=} \tilde g(x) & \; \Leftrightarrow \; \sum_i z_i \alpha_i {=} x, \, \sum_i z_i \beta_i {=} \tilde y, \, \sum_i z_i {=} 1, \, \sos(z_*) \qquad \label{eq:pwlconvex} 
\end{IEEEeqnarray}
\egroup

In addition, at most two adjacent basepoint weights are allowed to be non-zero, $z_i, z_{i+1} {>} 0$, $z_j {=} 0$, $j {<} i \vee i{+}1 {<} j$.
Most optimisation solvers allow to specify  such restrictions as  special order sets\footnotemark;
with such additional  information a solver's branch and bound methods can explore the search space more efficiently.
\footnotetext{$n$-Special Ordered Sets $\sos[n](V)$ are special constraints limiting decisions variables of a list $V$ so that at most $n$ adjacent variables are positive and non-zero having the remaining variables zero. Notation: $\sos[2](z_*)$ is a shorthand for $\sos[2](z_1,..,z_m)$ (\Cref{tbl:var}).}
With this restriction, coordinates could be expressed that are not located on line segments, \eg, the cross in \Cref{fig:intro1D}.%

Similar to univariate functions, bivariate functions $g(x,y){=}w$ can be linearised by triangles instead of segments~(\Cref{fig:intro2d}).
Then, the convex combination is extended by one parameter~\eqref{eq:pwlconvex2} and weights for at most three vertices of a single triangle are non-zero instead of the two segment vertices in the univariate version~\cite{MINLP2012} (two vertices of a segment correspond to two adjacent basepoints being non-zero).

\begin{samepage}
\bgroup\small
\begin{IEEEeqnarray}{ll}
\tilde w {=} \tilde g(x,y) \; \Leftrightarrow \; & \sum_i z_i \alpha_i {=} x, \, \sum_i z_i \beta_i {=} y, \notag \\
& \sum_i z_i \theta_i {=} \tilde w , \, \sum_i z_i {=} 1, \quad \sos(z_*) \qquad  \label{eq:pwlconvex2}  
\end{IEEEeqnarray}
\egroup\end{samepage}

The particular challenge here is to deal with an objective function having bivariate cost functions with one integer parameter~$y$.
The following two sections present different linearisation strategies for this challenge.

\subsection{Curves}\label{sec:lin:curve}

\begin{figure}[tbp]
\centering
\includegraphics[width=\columnwidth]{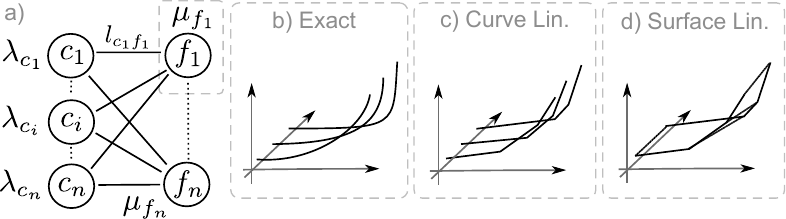}
\caption{Bipartite graph of a Facility Location Problem with queuing systems at each facility~(a) and zoomed bivariate time in system variants: exact~(b), curve-based~(c) and surface-based~(d).}
\label{fig:problemvis}
\end{figure}

Linearising \optQP boils down to linearising the non-linear part $\funcT(\lambda, y)$~\eqref{eq:tis} of the objective function~\eqref{eq:qpk:obj}.
But linearising $\funcT$ is not obvious as its second parameter is integer.
So, $\funcT$ is reformulated as $k_f$ separate univariate functions $\funcT[\mu,j](\lambda){=}\funcT(\lambda,j)\,$, $1{\leq}j{\leq}k_f$ (\Cref{fig:problemvis}b).
We use here the notation~$f_a(x)$ to indicate that $f$ is a function with constants $a$ and variables $x$, \eg $f_{5}(x), f_{7}(x)$ are two different  functions of $x$.
If exactly $y$~resources are allocated at a facility, function $\funcT[\lambda, j]$ computes the corresponding time in system.
Linearising all these functions (\Cref{fig:problemvis}c) also obtains a linearisation of the mixed-integer function $\funcT$~\eqref{eq:tis}.

Each of these functions~$\funcT[\mu,j]$ is convex, so we can use our algorithm~\cite{Keller2015} to obtain corresponding linearisations~$\funcTa[\mu,j]$ with high linearisation accuracy with few basepoints.
Assuming $m$~basepoints are used for one function, one facility needs $m \cdot n$~basepoints, $m\,k_f$.

The problem allows that  facilities have different service rates~$\mu$.
Linearising $\funcT[\mu,j]$ for different $\mu$ needs different basepoints for high linearisation accuracy. 
Consequently, in total $m \sum_f k_f$ basepoints\footnotemark{} exists.
\footnotetext{
Basepoint coordinates can be shared among facilities with same service rates. If all service rates are the same, only $m \, max_f k_f$ different basepoints exists.
}
The coordinates of a single basepoint~$(\alpha_{fji}, \beta_{fji})$ belong to facility~$f$ with $f{\in}F$, $1{\leq}j{\leq}n$, $1 {\leq} i {\leq} m$. They jointly define the approximation function $\funcTa[\mu,j]$.

We rewrite the original problem formulation.  First, we reformulate the integer variable~$y_f$ of decision problem \optQP into a vector of binary decision variables $\dot y_{fj}$, each one representing that facility~$f$ uses $j$~resources if and only if $\dot y_{f j} {=}1$.  Exactly one of $\dot y_{fj}$ is 1 ($\forall f {:}\; \sos[1](\dot y_{f*})\;$\footnotemark{}), which means that $y_f {=} \sum_j j\; \dot y_{fj}$.  Technically, $\dot y_{fj}$ is used to select which time in system function $\funcTa[\mu_f,j]$ for a \emph{specific} number $j$ of allocated reources is used; more formally, $\funcT[\mu_f](\lambda, k) {=} \sum_{j=1}^k \dot y_{fj} \funcT[\mu_f,j](\lambda)$.  \footnotetext{Notation: Shorthand $\dot y_{f*}$ for $\dot y_{f1},..,\dot y_{fn}$ (\Cref{tbl:var})}

Second, continuous weights~$z_{fi}$ are used to express the convex combination of basepoints~(\Cref{sec:pwl}), representing the utilisation and the corresponding TiS at each facility. 
For each facility~$f$, we need $m$~weights since only a single function~$\funcTa[\mu_f,j]$ is selected by the decision variable~$\dot y_{fj}$ (and not $m k_f$~weights, one weight per basepoint at each facility, as one might suspect). 
PWL function~$\funcTa[\mu_f,j](\lambda)$ is then formulated as $\sum_i z_{fi} \beta_{fji}$ with $\lambda_f$ as $\sum_i z_{fi} \alpha_{fji}$.

In the resulting problem formulation $\optQPaiii$\footnotemark{}~(\Cref{algo:optQPaiii}), $\optQP$'s objective function~\eqref{eq:qpk:obj} is transformed as described above, replacing~\eqref{eq:qpk:obj} by~\eqref{eq:qpka:obj} and by~\eqref{eq:qpka:ysos}.
\footnotetext{We use the following naming convention for optimisation problems: ``3'' indicates the cubic nature of this problem; the prefix ``c'' indicates the curve-based linearisation of the time-on-system function, later to be complemented by ``t'' for  triangular formulations and ``q'' for quadrilateral formulations.}
The new constraint~\eqref{eq:qpka:zsos} ensures a convex combination only of neighbouring weights.
The capacity constraint~\eqref{eq:qpka:cap} replaces~\eqref{eq:qpk:cap}.
The local and global limits are adjusted using $\dot y$; constraints~\eqref{eq:qpk:count}, \eqref{eq:qpk:limit} are replaced by \eqref{eq:qpka:count}, \eqref{eq:qpka:limit}.

\begin{optproblem}[tbp]
\caption{$\optQPaiii(G,\,p,\,\funcTa[**])$ \\ separate curves, cubic \label{algo:optQPaiii}}
\ifextended{}{\small}\vspace{-1mm}
\begin{IEEEeqnarray}{llls}
\min_{x,y,z} \, & \IEEEeqnarraymulticol{3}{l}{%
    \frac{\sum_{cf} x_{cf} l_{cf}} {\sum_c \lambda_c} {+} %
     \frac{\sum_{f} (\sum_c x_{cf}) 
         \overbrace{\sum_j \dot y_{fj}\sum_i z_{fi} \beta_{fji}
         }^{\text{time in system }f} } { \sum_c \lambda_c } %
    } \label{eq:qpka:obj} \\
\text{s.t.} \,  
&   \sum_f x_{cf} = \lambda_{c} & \forall c \qquad
    & {\small(demand)} \label{eq:qpka:demand} \\
&	\sum_c x_{cf} {\leq} \sum_j \dot y_{fj} \sum_{i} z_{fi} \alpha_{fji}  \quad & \forall f 
    & {\small(capacity)} \,\, \label{eq:qpka:cap} \\
& 	\sum_{i} z_{fi} = 1,\, \sos( z_{f*} ) & \forall f \,\,
    & {\small(weights)} \label{eq:qpka:zsos}\\
&	\sum_{j} j \, \dot y_{fj} \leq k_f &\forall f
    & {\small(count)} \label{eq:qpka:count}\\
&	\sum_{fj} j \, \dot y_{fj} = p  &   
    &{\small(limit)} \label{eq:qpka:limit}  \\
&	\sos[1]( \dot y_{f*} ) & \forall f 
    &{\small(curve flip)} \label{eq:qpka:ysos}
\end{IEEEeqnarray}
\vspace{-4mm} 
\end{optproblem}

The right term of the objective function~\eqref{eq:qpka:obj} multiplies three decision variables, $x\dot yz$, making the problem cubic. 
To transform it into a linear problem, at first, the factor of \optQPaiii's time-in-system function~\eqref{eq:qpka:obj} is integrated in the basepoint coordinates. 
To understand this modification, \optQP's objective function~\eqref{eq:qpk:obj} has to be revisited.
In this function, the time-in-system function~$\funcT$ has the factor $\sum_c x_{cf}$, in short $u$.
This term of the objective function can be reformulated as $u \funcT(u,y)$. 
As $u$ is also a parameter of $\funcT$, the term also equals $\funcN(\nicefrac{u}{\mu}, y)$ (\funcN defined in \eqref{eq:nis}).
This term has two advantages: 
Without having a factor, the objective function becomes simpler (from~\eqref{eq:qpka:obj} to~\eqref{eq:qpka:obj2}).
In contrast to \funcT, \funcN is independent of $\mu$ and basepoint coordinates need only computed once and not for each different $\mu$. 
{\vspace{-1mm}\small
\begin{IEEEeqnarray}{llls}
\min_{x,y} \, & \IEEEeqnarraymulticol{3}{l}{%
    \frac {1} {\sum_c \lambda_c}%
    \sum_{cf} x_{cf} l_{cf} +%
    \frac {1} {\sum_c \lambda_c}%
    \sum_{f}  \funcN( \frac{\sum_c x_{cf}}{\mu_f},\, y_{f} )%
    }%
    \label{eq:qpk:obj2}
\end{IEEEeqnarray}
}

To linearise $\funcN$, $\funcN( \nicefrac{u}{\mu_f} ,y)$ is split into functions~$\funcN[_j]( \nicefrac{u}{\mu_f} ) {=} \funcN( \nicefrac{u}{\mu_f} ,j)$ separately being linearised (similar to $\funcT$).
The inner function~$\nicefrac{u}{\mu_f}$ of the abscissa is resolved by modifying the abscissa coordinates of the linearisation basepoints.
The resulting convex combination is defined by \eqref{eq:qpka:nis1} (with weights~$z_i$).
{\vspace{-1mm}\small
\begin{IEEEeqnarray}{l}
\funcNa(\nicefrac{u}{\mu}): 
\begin{cases}
\sum_{i} z_{i} \alpha_{i}= \nicefrac{u}{\mu}  \, \Leftrightarrow \,  \mu \sum_{i} z_{i} \alpha_{i} = u \\
\sum_{i} z_{i} \beta_{i} = \funcN( \nicefrac{u}{\mu} )
\end{cases} \hspace*{-2mm} , u {=} \sum_c x_{cf} \quad  \label{eq:qpka:nis1} 
\end{IEEEeqnarray}
}

Replacing \optQPaiii's object function~\eqref{eq:qpka:obj} with \eqref{eq:qpka:obj2} and capacity constraint~\eqref{eq:qpka:cap} with \eqref{eq:qpka:cap2} yields optimisation problem ~\optQPaii (now quadratic, hence a ``2'').
{\vspace{-1mm}\small
\begin{IEEEeqnarray}{llls}
\min_{x,y,z} \, & \IEEEeqnarraymulticol{3}{l}{%
        \frac{\sum_{cf} x_{cf} l_{cf}} {\sum_c \lambda_c} {+} %
         \frac{\sum_{fj} \dot y_{fj}\sum_i z_{fi} \beta_{ji} } { \sum_c \lambda_c } %
        } \label{eq:qpka:obj2} \\
&	\sum_c x_{cf} {\leq} \mu_f \sum_j \dot y_{fj} \sum_{i} \alpha_{ji} z_{fi} \, & \forall f \quad
    &  \small(capacity) \quad \,\, \label{eq:qpka:cap2}        
\end{IEEEeqnarray}
}

In consequence, \optQPaii needs fewer basepoints than \optQPaiii, because $\funcNa_j$ depends only on the number~$m$ of allocated resources~$j$ but is now independent of $\mu$; in total only $m \cdot \max_f k_f$~basepoints are necessary.

The objective function of \optQPaii~\eqref{eq:qpka:obj2} is quadratic ($\dot yz$). 
We further simplified the problem by two modifications to derive a linear problem formulation:
Firstly, more weights~$z$ are used; previously only $m$ weights per facility $f$ are modelled whereas now $mn$ weights are modelled:  $z_{fji}$ for facility $f$,  $\funcNa_j$, basepoint $i$.
Doing so yields an equivalent problem with increased search space.
Secondly, $\dot y$ is turned into a constraint enforcing that only those weights $z_{fji}$ are non-zero which correspond to $\funcNa_j$ being active at facility $f$, $\dot y_{fj}$.
Several test runs showed that the resulting linear problem is solved faster than its quadratic counterpart despite its larger search space.

\begin{optproblem}[tbp]
\caption{$\optQPai(G,\,p,\,\funcNa_*)$\\ separate curves, linear\label{algo:optQPai}}
\ifextended{}{\small}\vspace{-1mm}
\begin{IEEEeqnarray}{llls}
\min_{x,y,z} \, & \IEEEeqnarraymulticol{3}{l}{%
    \frac{\sum_{cf} x_{cf} l_{cf}} {\sum_c \lambda_c} {+} %
     \frac{ 
         \overbrace{\sum_{fji} z_{fji} \beta_{ji}
         }^{\text{total time in systems}} } { \sum_c \lambda_c } %
    } \quad \label{eq:qpka2:obj} \\
\text{s.t.} \,  
&   \sum_f x_{cf} = \lambda_{c} & \forall c \quad
    & {\small(demand)} \label{eq:qpka2:demand} \\
&	\sum_c x_{cf} {\leq} \mu_f \sum_{ji} z_{fji} \alpha_{ji}  \, & \forall f 
    & {\small(capacity)} \hspace{1.5em} \label{eq:qpka2:cap} \\
& 	\sum_{i} z_{fji} = 1,\, \sos( z_{fj*} ) \, & \forall f,j \,\,
    & {\small(weights)} \label{eq:qpka2:zsos}\\
&   \sum_i z_{fji} = \dot y_{fj}  & \forall f,j \quad
    &{\small(sync)} \label{eq:qpka2:sync} \\
&   \IEEEeqnarraymulticol{2}{s}{Constraints \eqref{eq:qpka:count}, \eqref{eq:qpka:limit}, \eqref{eq:qpka:ysos} } \notag
\end{IEEEeqnarray}
\vspace{-4mm}
\end{optproblem}

The resulting linear formulation $\optQPai$~(\Cref{algo:optQPai}) has a new constraint~\eqref{eq:qpka2:sync}, ensuring that only the relevant weights are allowed to be non-zero. 
Having this constraint in place, the new objective function~\eqref{eq:qpka2:obj} equals the old objective function~\eqref{eq:qpka:obj2}.
Constraints~\eqref{eq:qpka2:cap}, \eqref{eq:qpka2:zsos} now support the new $j$ index for weights $z$.
In this way, the new linear problem~\optQPai computes the same solution as the previous quadratic problem~\optQPaii.

\subsection{Thinned Curves}  \label{sec:lin:curvethin}

Problem~\optQPai uses $mn$~basepoints at each facility, $m$ being chosen and $n{=}k_f$.
For example, for $100$~facilities, each with $40$~resources available, \optQPai has $n{=}40$ separate functions $\funcNa_j$ at each facility. 
When using $m{=}10$ basepoints for each PWL function~$\funcNa_j$, the search space contains \num{40000} weight decision variables.
The search space can be reduced by reducing $m$ or $n$, but this lowers accuracy.
Reducing the number of basepoints is a straightforward trade-off of accuracy against search space size, 
but simply reducing the number of resources~$n$ is not adequate as this modifies the problem instance.
Hence, to obtain a similar trade-off for the resources, we need to find a way to reduce the amount of resources to look at.  One option would be to allow gaps in the sequence of number of allocated resources, with appropriate rounding, $J{=}[1,k_f]$.
For example, with $40$ available resources at a facility, we could remove the option to use, say, $26$~resources, forcing the solution to either use $25$ or $27$ resources instead; this turns the interval $J$ into an explicitly enumerated set:  $J{=} \{1,..,25,27,..,k_f\}$.
In the example with $100$ facilities and $m{=}10$ basepoints this removes $1000$ decision variables.
The set~$J$ specifies which options $j{\in}J$ for the number of allocated resources at one facility are available.
\begin{figure}[tb]
\centering
\ifextended{
    \includegraphics[scale=0.9]{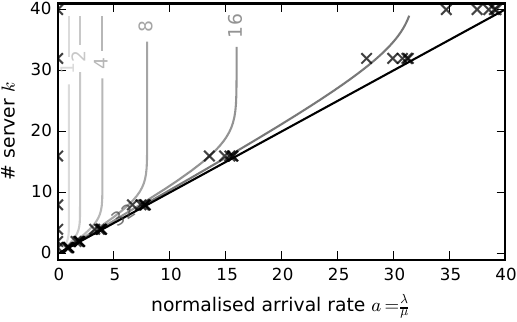} }{
    \includegraphics[scale=0.7]{work/intro_tis_contour_nis} }
\caption{
Contour plot for $\funcN_j(a)$ with basepoints for $j{\in}J{=}\{1,2,4,8,16,32,40\}$ of $\funcNa_J(a,y)$.
\label{fig:contournis}
}

\end{figure}
As an example, \Cref{fig:contournis} shows $\funcN(a,j){=}t$ as contour lines.
Basepoints are plotted as crosses for $\funcNa_k(a)$, $j{\in}J=\{1,2,4,8,16,32,40\}$ with $|J|{=}n{=}7$, $m{=}6$.
With $|F|{=}100$ facilities, only $|F|mn{=}\num{4200}$ weight decision variables are needed as opposed to \num{40000} variables with $n{=}40$ and $m{=}10$.
Each PWL function~$\funcNa_j$ can be seen as a separate curve and thinning the sequence of available functions can be seen as thinning the list of all curves -- these two perceptions gave the name for these formulation techniques.

\begin{optproblem}[tbp]
\caption{$\optQPa(G,\,p,\,\funcNa_{*},\, J)$, thinned curves \label{algo:optQPa}}
\vspace{-1mm}
\begin{IEEEeqnarray}{llls}
\IEEEeqnarraymulticol{4}{s}{%
Obj.$\,$func./Constr.~\eqref{eq:qpka2:obj},\eqref{eq:qpka:ysos},\eqref{eq:qpka2:demand},\eqref{eq:qpka2:cap},\eqref{eq:qpka2:zsos},\eqref{eq:qpka2:sync} }\notag  \\
\text{s.t.} \quad
&	\sum_{j{\in}J} j \; \dot y_{fj} \leq k_f \quad &\forall f \qquad
    & {\small(count)} \label{eq:qpka3:count}\\
&	\sum_{f,\, j{\in}J} j \; \dot y_{fj} \leq p  &   
    &{\small(limit)} \label{eq:qpka3:limit}  
\end{IEEEeqnarray}
\vspace{-2mm}
\end{optproblem}

Problem~$\optQPai$~(\Cref{algo:optQPa}) is adjusted by adding the set  $J$ as parameter and replacing the constraints~\eqref{eq:qpka:count}, \eqref{eq:qpka:limit} by \eqref{eq:qpka3:count}, \eqref{eq:qpka3:limit}.
By dropping a particular $j$ from $J$, two issues can arise: 
a)~In special cases the problem becomes infeasible (as the options to represent a solution are reduced).
b)~The number of allocated resources $y_f {\in} J,\, y_f {=} \sum_{j{\in}J} j \; \dot y_{fj} $ cannot satisfy $\sum_f y_f {=} p$; 
this is the reason why \Coref{eq:qpka3:limit} is relaxed to an upper bound.

\begin{inextended}
To illustrate both issues, consider the following example: 
Resources should be allocated at $6$ facilities, deciding values for $y_{1..6}$. 
The facilities have $k_f{=}30$ resources available and all resources are homogeneous.
Assume that the demand is handled by $p{=}18$ resources 
and that  at least $15$~resources are needed, $\lceil \frac {\sum_c \lambda_c}{\mu} \rceil {=} 15$.
Basepoints exist for $\forall j {\in} J : \funcNa_j(a)$.
Starting with all resource allocations available, $J{=}[1,30]$, the optimal allocation is $y_{1..6} = 3 {\in} J$.
A thinned $J{=}[1,30]$ renders \optQPa infeasible.
Allocating $30$~resources at one site violates the upper limit, $30 {\not =} 18 {=} p$.
Similarly, allocating one resource at each facility violates the limit, $6{\not =}18{=}p$.
Removing fewer $j$ from $J$ reduces the likeliness of this special case.

To illustrate issue b) from above, consider \optQPa  with
$J{=}\{1,10,20,30\}$. Under constraint~\eqref{eq:qpka3:limit}, this problem
is feasible, \eg, with allocation $y_{1..6}{=}10,1,1,1,1,1$ which
satisfies the limit $\sum_f y_f {=} 16 \leq p {=} 18$ and also
satisfies all demand (it would not be feasible if
constraint~\eqref{eq:qpka3:limit} stipulated equality).

\end{inextended}

However, the downside is that not all $p$~resources are used and
queuing delay could be reduced further since 
adding a resource at any facility always reduces that facility's queuing delay.
So adding the remaining resources, $\sum_f y_f {-} p$, will always improve the solution.
The algorithms~\Call{Alloc}{} and \Call{MaxCostDrop}{} increase
resources at those facilities where the queuing delay is reduced
the most. 
They allocate the remaining resources in a way that optimally improves the solution.
The auxiliary algorithm~\Call{Alloc}{} merely transforms many
\optQP variables into a structure suitable for the generic greedy
algorithm~\Call{MaxCostDrop}{}.
Both algorithm are presented here in a more generic form as both are also used later.
\Call{MaxCostDrop}{} considers facilities as buckets ($f$) and allocating resources as placing tokens ($y'_f$) into these buckets; $n{=}\sum_f y_f {-} p$ times.
The queuing delay drop of allocating more resources ($N(a_f, y_f{+}y''_f)$) is described by a decreasing, convex cost function~($c_f(y''_f)$), \Cref{algo:alloc} \Lineref{algo:alloc:costf}.
\Call{MaxCostDrop}{} successively places one token at the bucket with the resulting, largest cost drop at this step.
Distributing tokens to buckets in this way maximises the total cost and the queuing delay drop (\Cref{theo:maxcostdrop}).

\begin{algorithm}[tbp]
\caption{\textsc{Alloc}$(p,\, \lambda',\, y{=}0,\, F'{=}F ) \; \rightarrow \; y, T$ \\ Optimal allocation for known assignment. } \label{algo:alloc}%
\centering
\ifextended{}{\small}
\begin{algorithmic}[1]
\Statex \hspace{-5mm} \textbf{requires} $\lambda'_f {<} k_f\mu_f, \sum_f \lceil \nicefrac{\lambda'_f}{\mu_f}\rceil {\leq} p $
\Statex \hspace{-5mm} \textbf{ensures} $\lambda'_f {<} y_f\mu_f, \sum_f y_f {=} p $
\Statex \hspace{5mm} minimises $\sum_f \funcN( \nicefrac{\lambda'_f}{\mu_f}, y_f )$
\State $\forall f {\in} F': a_f {\gets} \nicefrac{\lambda'_f}{\mu_f}, \; \; y_f^{\text{min}} {\gets} \max \{ y_f, \, \lceil a_f \rceil \} $ \label{algo:alloc:minalloc}
\State $n \gets p{-} \sum_f y_f^{\text{min}}$
\If{$n {>} 0$}
    \State $y' \gets \Call{MaxCostDrop}{n,\, c(j) ,\, y^{\text{max}} }$ \label{algo:alloc:costf}
    \Statex \hspace{5em} $\forall f: c_f(j) := \funcN(a_f, y_f{+}j)$   
    \Statex \hspace{5em} $\forall f: y_f^\text{max} \gets k_f {-} y_f^\text{min}$
    \State $\forall f : y_f =  y_f^\text{min} + y'_f$ \label{algo:alloc:costs}
\Else
    \State $\forall f : y_f =  y_f^\text{min}$
\EndIf
\State \Return $y ,\, \sum_f \funcN(a_f, y_f)$
\end{algorithmic}
\hrule \raggedright
\small $y_f^\text{min}$: minimal necessary allocation at $f$ \\
$y_f^\text{max}$: remaining resources usable by \Call{MaxCostDrop}{} \\
$c_f(j)$: weighted time in system at $f$ with $j$  
\end{algorithm}

\begin{algorithm}[tbp]
\caption{\textsc{MaxCostDrop}$(n ,\, c(j),\, y^{\text{max}}) \; \rightarrow \; y$\\ 
    Drops $n$~tokens in $m$~buckets while minimising total drop costs under bucket capacity constraint.\label{algo:mincostdrop}}
\centering
\ifextended{}{\small}
\begin{algorithmic}[1]
\Statex \hspace{-5mm}  \textbf{requires} $\forall f{:}\; c_f(x)$ is decreasing and convex
\Statex \hspace{-5mm}  \textbf{ensures} $\forall f{:}\; y_f {<} y_f^{max}$, $\sum_f y_f {=} n$,
max $\sum_f c_f(y_f)$
\State $S \gets \left\lbrace y_{fi}  \,\vert\, f{\in}F, 1{\leq}j{\leq}y_f^{\text{max}} \right\rbrace$ \label{line:mincostdrop:matroid}
\Statex \hspace{1.5em} $I \gets \left\lbrace S' {\subseteq} S  \,\vert\, n {\geq} |S'| \right\rbrace$
\Statex \hspace{1.5em} $w(y_{fj}) := c_f(j)-c_f(j{-}1)$ \textbf{if} $y{>}1$ \textbf{else} $c_f(0)$ 
\State $A \gets \emptyset$ \label{line:mincostdrop:greedyb}
\For{$y_{fj} \in S$, sorted by non-increasing weights} \label{line:mincostdrop:greedysort}
    \If{$|A| < n$} $|A| \gets A \cup \lbrace y_{fj} \rbrace$
    \EndIf
\EndFor \label{line:mincostdrop:greedye}
\State $\forall f{:}\; y_f \gets \vert \lbrace y_{fj} {\in} A \,\vert\, 1 {\leq}j{\leq}y_f^{\text{max}} \rbrace \vert$ \label{line:mincostdrop:tox}
\State \Return $y$
\end{algorithmic}
\hrule \raggedright \small 
$y_{fj}$: the $j$-th token placed in $f$; $\;$
$w_{fj}$: the cost for placing $y_{fj}$ 
\end{algorithm}

\begin{theorem}
\label{theo:maxcostdrop}
\Call{MaxCostDrop}{$n ,\, c_*(y),\, y^{max}$}~(\Cref{algo:mincostdrop}) maximises $\sum_f c_f(y_f)$ ensuring $\forall b{:} y_f < y_f^{max}$, $\sum_f y_f {=} n$ for any decreasing and convex cost function $c_f(x)$.
\end{theorem}
\vspace{-2mm}
\begin{proof}
For a weighted matroid $M{=}(S,I)$, algorithm
\Call{Greedy}{} (from~\cite{Cormen2009}, Theorem~16.10, p.~348ff) computes a subset$A$ with maximal weight.
\Lineref{line:mincostdrop:matroid} defines a weighted matroid $M{=}(S,I)$: $S$ is non-empty; $I$ is hereditary meaning $\forall B{\in}I, A{\subseteq}B: A{\in}B$; $M$ satisfies the exchange property meaning $\forall A,B{\in}I, |A|{<}|B|: x {\in} B{-}A \,\wedge\, A {\cup} \lbrace y \rbrace {\in} I$; $w(y)$ is positive.
The lines \eqref{line:mincostdrop:greedyb}--\eqref{line:mincostdrop:greedye} are Corman's \Call{Greedy}{} casted to our $M$.
By adding $y_{fj}$ to subset $A$, one token is added to bucket~$f$; so the number of tokens~$y_f$ in each bucket can be aggregated as in line~\eqref{line:mincostdrop:tox}. 
In particular, the following property holds: $\forall f,j {:} y_{fj}$ was added before $y_{f(j+1)}$.
This is ensured by the weight sorting (line~\eqref{line:mincostdrop:greedysort}) and $c_f(y)$ being decreasing and convex; then $\forall f,j {:} w_{fj} {>} w_{f(j+1)}$.
From the same property the costs can be derived as $\forall f{:}\; c_f(y_f) {=} \sum_{y_{fj}{\in}A} w(y_{fj})$. 
From Theorem~16.10, the computed subset $A{\subset}S$ maximises $\sum_{y_{fj}{\in}A}w(y_{fj})$ which also maximises $\sum_b c_f(y_f)$.
\end{proof}

\begin{notinextended}
\Call{Greedy}{} places tokens maximising the placement costs and is used by algorithm~\Call{Alloc}{} to allocate $n$~remaining resources to maximally decrease the corresponding response time. The paper's extended version~\cite{Keller2015b} contains a detailed proof.
\end{notinextended}

\begin{inextended}

\begin{lemma}
\label{theo:alloc}
\Call{Alloc}{$G,\, \lambda',\, p$}~(\Cref{algo:alloc}) allocates $p$ resources at facilities $f{\in}F'$ so that $\sum_f \funcN(\nicefrac {\lambda'_f}{\mu_f}, y_f){=}T$ is minimised while ensuring $\lambda'_f {<} y_f\mu_f$ and $\sum_f y_f {=} p$.
\end{lemma}
\vspace{-2mm}
\begin{proof}
Line~\eqref{algo:alloc:minalloc} allocates the resources $y_f^\text{min}$ at least necessary to handle assigned demand $\lambda'_f$ having $n{\geq}0$ resources still to be allocated ($n{<}0$ contradicts $\lambda'_f$'s requirements).
With $n{=}0$ the computed minimal allocation is the only one and, hence, $T$ is minimal; done.
With $n{>}0$ greedy algorithm~\Call{MaxCostDrop}{} is used, where placing tokens into buckets corresponds to increasing resource at facilities.
The token cost function $c_f(j)$ and $N_a(j)$ is decreasing and convex in $j$ and $y^{max}_f$ specifies the capacity of bucket $f$.
Placing $n$~tokens in buckets maximises the total token cost~$\sum_f c_f(y_f)$ (\Cref{theo:maxcostdrop}).
The token cost corresponds to the reduction, $\forall f{:}\; N_a(y_f^{min}) {-} N_a(y_f^{min}{+}y''_f) = c_f(y''_f)$, incurred by allocating $y''_f$ additional resources at $f$.
With minimal allocation ($y^\text{min}$, line~\eqref{algo:alloc:minalloc}) $T$ is maximal. 
Adding resources with maximal reduction of token cost also maximally reduces $T$. 
Consequently, the resulting $T$ is minimal; done.
\end{proof}

\end{inextended}

\subsection{Surface with Triangles}\label{sec:lin:surfacetriangle}

\begin{inextended}
\begin{figure*}[tbp]
\centering
\includegraphics[width=0.8\linewidth]{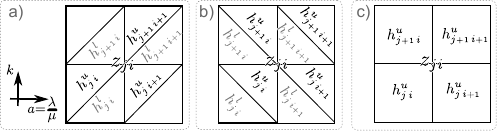}
\caption{
Surface linearisations with MILP formulation relationship; a) \optQPb, b) \optQPbb, and c) \optQPc.\label{fig:surfacemodel}
}
\end{figure*}
\end{inextended}
\begin{notinextended}
\begin{figure}[tbp]
\centering
\includegraphics[width=1.0\columnwidth]{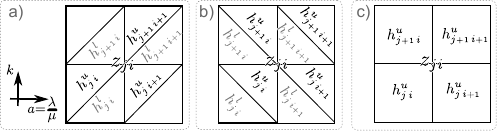}
\caption{
Surface linearisations with MILP formulation relationship; a) \optQPb, b) \optQPbb, and c) \optQPc.\label{fig:surfacemodel}
}
\end{figure}
\end{notinextended}

As discussed in the previous section, dropping separate curves from $J$ jeopardizes feasibility.
To overcome this issue, the previously separate univariate PWL functions are now joined into one bivariate PWL function.
This is done by creating a mesh of triangles over all basepoints of all curves in $J$ (\Cref{sec:pwl}); 
such a mesh defines also the surface of the linearisation (\Cref{fig:problemvis}d).
Doing so, the problem of reduced allocation options (thinned curves) disappears as they are implicitly represented by interpolating between basepoints of neighbouring PWL functions.
Hence, the set of basepoints can be thinned out more aggressively to further reduce the search space without jeopardising feasibility.

This section discusses three questions: 
1)~Using neighouring functions' basepoints for interpolation introduces inaccuracy -- how large is it?
2)~Modelling triangles is more complex than modelling line segments -- will a smaller search space compensate for a more complex optimisation problem?
3)~The convex combination  formulation introduces continuous weights, rendering $y_f$  continuous -- how to treat fractional server allocations?

The linearisation surface $\tilde S$ (the mesh of triangles between basepoints of sequence~$J$) linearly approximates the surface~$S$ of the original function $g$, which is in our case $\funcN(a,y)$; \Cref{fig:contournis} shows its contour for $1{\leq}y{\leq}40$ and corresponding basepoints of an example $J$. 
The triangle mesh (\Cref{fig:surfacemodel}a) has $mn$ basepoints.
As before, the basepoints~$\alpha_{ji}$, $\beta_{ji}$, $\theta_{ji}$ touch the surface at $g(\alpha_{ji}, \beta_{ji}) {=} \theta_{ji}$, $0{\leq} j {<} n$, $0 {\leq} i {<} m$.
Using the convex combination MILP formulation~\eqref{eq:triangle1}, each basepoint has a weight~$z_{ji}$ as its decision variable. Only the three edges of exactly one triangle form the convex combination (\sos[3])~\cite{MINLP2012,DAmbrosio2010}. This is naturally expressed by \sos[3], but these are not supported by current solvers. 
D'Ambrosio\etalcite{DAmbrosio2010} presented an equivalent formulation~\eqref{eq:triangle2} only using \sos[1] constraints:
A new auxiliary, binary decision variable $h$ for each triangle indicates whether the triangle is active, only one $h$ and its triangle weights are allowed to be non-zero while all other weights are forced to zero.
The formulation~\eqref{eq:triangle2} corresponds to triangle enumeration and orientation in~\Cref{fig:surfacemodel}a.
Integrating the triangle formulation from~\eqref{eq:triangle1}, \eqref{eq:triangle2} into \optQPa results in $\optQPb$ (\Cref{algo:optQPb}).
\bgroup \small
\begin{IEEEeqnarray}{l}
\sum_{ji} \alpha_{ji} z_{ji} {=} x,\, \sum_{ji} \beta_{ji} z_{ji} {=} y, \sum_{ji} \theta_{ji} z_{ji} {=}w, w = g(x,y), \notag \\
\sos[3](z_{**}),\,  \sum_{ij} z_{ij} = 1 \label{eq:triangle1} \\
\sum_{ji}{ h^u_{ji} {+} h^l_{ji} } {=} 1, \, \sos[1]( h^{*}_{**} ),\, 
    h^*_{0*} {=} h^{*}_{*0} {=} h^{*}_{m*} {=} h^{*}_{*n} {=} 0  \notag  \\
\forall ji: z_{ji} \leq h^u_{ji} {+} h^l_{ji} {+} h^u_{\underset{j\hspace{3mm}}{i+1}} 
        {+} h^l_{\underset{j+1}{i+1}} {+} h^u_{\underset{j+1}{i+1}} {+} h^l_{\underset{j+1}{i\hspace{3mm}}} \label{eq:triangle2}
\end{IEEEeqnarray}
\egroup

While D'Ambroisio\etal focus on tight approximation using many basepoints, we additionally want to reduce solving time. This depends on the number of used decision variables and basepoints.
So in a nutshell, we aim for large triangles with high linearisation accuracy; the accuracy is, as usual, the maximal difference between the original surface~$S$ and the triangle's surface~$\tilde S$.
This is facilitated by  choosing good basepoints, and the remaining section discusses two further options:  changing triangle orientation (this subsection) and replacing the triangles with quadrilaterals (Section~\ref{sec:lin:surfacequad}).

\begin{optproblem}[tbp]
\caption{$\optQPb(G,\,p,\, \funcNa)$, with triangles\label{algo:optQPb}}
\vspace{-1mm}
\ifextended{}{\small}
\begin{IEEEeqnarray}{llls}
\min_{\begin{smallmatrix}
x,y,\\
z,h
\end{smallmatrix}} \, 
& \IEEEeqnarraymulticol{3}{l}{%
    \frac{1} {\sum_c \lambda_c} \sum_{cf} x_{cf} l_{cf} {+} %
     \underbrace{\frac{1
         } { \sum_c \lambda_c } \sum_{fji} z_{fji} \theta_{ji} }_{\text{sum of all TiS} } %
    } \quad \label{eq:qpkb:obj} \\
\text{s.t.} \,  
&   \sum_f x_{cf} = \lambda_{c} & \forall c \quad
    & {\small(demand)} \label{eq:qpkb:demand} \\
&	\sum_c x_{cf} {\leq} \underbrace{\mu_f \sum_{ji} z_{fji} \alpha_{ji}}_{\text{capacity at }f} \, & \forall f 
    & {\small(capacity)} \hspace{1.5em} \label{eq:qpkb:cap} \\
&	\underbrace{ \sum_{ji} z_{fji} \beta_{ji}}_{\text{\#\,server used at }f} \leq k_f & \forall f
    & {\small(count)} \label{eq:qpkb:count}\\    
&	\sum_{fji} z_{fji} \beta_{ji} = p  &   
    &{\small(limit)} \label{eq:qpkb:limit}  \\
&	\sum_{ji} z_{fji} = y_f & \forall f 
    &{\small(open)} \label{eq:qpkb:open} \\
&   \sum_{j'i'} h^{*}_{fj'i'} = y_f  & \forall f \quad
    &{\small(o-sync)} \label{eq:qpkb:open2} \\
&   \sum_{j'i'}{ h^u_{fj'i'} + h^l_{fj'i'} } = 1, \notag \\
&   \hspace{3em} \sos[1]( h^{*}_{f**} ) & \forall f
    & {\small(single-tri)} \label{eq:qpkb:singletri} \\   
&   h^*_{f0*} {=} h^{*}_{f*0} {=} h^{*}_{fm*} {=} h^{*}_{f*n} {=} 0 \qquad &   \forall f
    & {\small(tri-corner)}\,\, \label{eq:qpkb:tricorner} \\
&   \IEEEeqnarraymulticol{3}{l}{%
    z_{ji} \leq h^u_{fji} + 
    h^l_{fji} + h^u_{f\,j\,i+1} + 
    h^l_{f\,j+1\,i+1}+}\notag\\ 
&   \hspace*{3em}        h^u_{f\,j+1\,i+1} + h^l_{f\,j+1\,i}  & \forall ji
    & {\small(tri-sync)} \label{eq:qpkb:trisync} \\           
&   \text{with } 0{\leq}j'{\leq}n ,\, 0{\leq}i'{\leq}m \notag
\end{IEEEeqnarray}
\end{optproblem}

The resulting, triangle-based optimisation problem has four groups of decision variables: 
a)~The request assignment $x_{cf}$; 
b)~binary variable $y_f$ activates facility~$f$ and considers time in system only for active facilities;
c)~the weights~$z_{fji}$ of the linearised surfaces at each facility $f$;
d)~the auxiliary variables~$h^u_{fji}$, $h^l_{fji}$ for the upper and lower triangles ensure that weights~$z_{fji}$ are 
non-zero if exactly one triangle at each facility is active.

As a variant, the triangle direction can be flipped (\Cref{fig:surfacemodel}b).
The adjusted formulation $\optQPbb(G,\,p,\, \funcNa)$ uses objective function and constraints from \eqref{eq:qpkb:obj}--\eqref{eq:qpkb:tricorner} and replaces constraint~\eqref{eq:qpkb:trisync} by \eqref{eq:qpkb2:trisync}.
{
\small
\begin{IEEEeqnarray}{lls}
 \IEEEeqnarraymulticol{3}{l}{%
        z_{ji} \leq h^u_{fji} + h^l_{f\,j\,i+1} + h^u_{f\,j\,i+1} + h^l_{f\,j+1\,i+1}+} \notag\\ 
   \hspace*{3em}        h^u_{f\,j+1\,i} + h^l_{f\,j+1\,i}\quad  & \forall ji\quad
    & {\small(tri-sync)}\qquad\qquad \label{eq:qpkb2:trisync} 
\end{IEEEeqnarray}
}

Comparing the linearisations of both triangle orientations, the resulting surfaces $\tilde S^{+}, \tilde S^{-}$ are obviously different.
Zooming into two adjacent triangles, \Cref{fig:convcomb4points} shows four example basepoints and the two triangles of each of the two orientations. 
The three visible triangles are distinguished by different shades of grey (the fourth one is hidden in the back).
Two triangles connected by the dashed diagonals either form an upper $\tilde S^{+}$ or lower surface $\tilde S^{-}$, depending on the orientation and on the differences of the basepoint coordinates. 
 
\begin{figure}[tbp]
\centering
\begin{minipage}[b]{0.48\columnwidth}
  \centering
  \includegraphics[width=0.9\columnwidth]{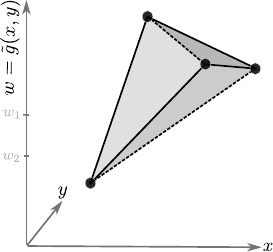}
  \caption{Triangle surface between four basepoints.\label{fig:convcomb4points}}
\end{minipage}
\hfill
\begin{minipage}[b]{0.48\columnwidth}
  \centering
  \includegraphics[width=1.0\columnwidth]{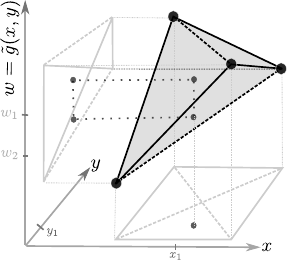}
  \caption{Quadrilateral surface; convex hull of convexly combining four basepoints.\label{fig:convcomb4pointsquad}}
\end{minipage}

\end{figure}

All basepoints lie on the original surface~$S$. 
At any other points, at a given ($x, y$) coordinate, there are two
points $(x, y, w) \in S$ (the original point) and  $(x, y, \tilde w) \in \tilde S$ (its linearisation);
typically, $w \not = \tilde w$, inducing some linearisation
inaccuracy. 
In general, this inaccuracy depends on the triangle orientation, the number of used basepoints, and the basepoint coordinates itself.
\begin{inextended}
Taking this fact into account, \Cref{sec:eval:bps} discusses the
generation of basepoint sets while maximising accuracy.
\end{inextended}
\begin{notinextended}
We discuss  the  choice of basepoints maximising accuracy in the extended version.  
\end{notinextended}

\subsection{Surface with Quadrilaterals}\label{sec:lin:surfacequad}
Triangle-based linearisations can be found in the literature.
We explore here an alternative, namely a linearisation where the basepoints form quadrilaterals rather than triangles. 
The hope is again to reduce the search space.

For linearising a function $g(x,y){=}w$ a single 
triangle can be described by three basepoints and all points within the triangle by convexly combining the basepoints~\eqref{eq:rt:tri}.
Similarly,  the four basepoints of a quadrilateral can be used~\eqref{eq:rt:quad}.
However, while equations~\eqref{eq:rt:tri} form a linear system having
a unique solution, the corresponding
equations~\eqref{eq:rt:quad} form an underdetermined linear system having more unknowns ($z_i$) than equations.
\bgroup
\small
\begin{IEEEeqnarray}{l}
\sum_{i{=}0}^{2}\alpha_i z_i {=} x,\, \sum_{i{=}0}^{2}\beta_i z_i {=} y, \,\sum_{i{=}0}^{2}\theta_i z_i {=} w \quad \label{eq:rt:tri} \\
\sum_{i{=}0}^{3}\lambda_i z_i {=} x,\, \sum_{i{=}0}^{3}\beta_i z_i {=} y,\, \sum_{i{=}0}^{3}\theta_i z_i {=} w \quad \label{eq:rt:quad} 
\end{IEEEeqnarray}
\egroup

This results in the following issue: For given $x,y$-coordinates,
there are infinitely many solutions for $w$-coordinates admissible
under equations~\eqref{eq:rt:quad}.
For a geometric illustration, \Cref{fig:convcomb4pointsquad} shows four basepoints. 
For a fixed $(x_1, y_1)$, the admissible values for $w$ are visualised by a
vertical grey line.  This line intersects with the basepoints' convex
hull at $(x_1, y_1, w_1)$ and $(x_1, y_1, w_2)$, limiting admissible $w$ to
the interval~$[w_1, w_2]$.

In general, this formulation is difficult to integrate into an MILP because
$w$ is not unique.  However, for problems minimising $w$, the optimum
is $w_1$; maximising $w$ results in $w_2$.  For such problems, $w$
values are unique.  
And indeed, the relevant term of our objective function~\eqref{eq:qpkb:obj}  $\sum_{f} \; \sum_{ji}z_{fji}\theta_{ji}$ is to be minimised.
Hence, the quadrilateral formulation is  applicable.

The resulting problem formulation $\optQPc(G,\,p,\,\funcNa)$  has objective function and constraints from \eqref{eq:qpkb:obj}--\eqref{eq:qpkb:tricorner} and replaces constraints~\eqref{eq:qpkb:trisync} by \eqref{eq:qpkb23:trisync}.
\begin{IEEEeqnarray}{lls}
      z_{ji} \leq h_{fji} + h_{f\,j\,i+1} + \notag \\
 \hspace*{3em}   h_{f\,j+1\,i+1} + h_{f\,j+1\,i} & \quad \forall ij
    & \quad {\small(rec-sync)} \label{eq:qpkb23:trisync} 
\end{IEEEeqnarray}
The advantage of this approach is that it needs fewer decision
variables: We need only half the number of quadrilaterals to cover an
area than triangles (\Cref{fig:surfacemodel}), and each triangle or
each quadrilateral needs its own decision variable to control whether
its basepoints contribute to the convex combination. Moreover, the
constraints in the quadrilateral case are simpler than in the triangle
case -- compare \eqref{eq:qpkb:trisync} vs.\
\eqref{eq:qpkb23:trisync}.

\begin{algorithm}[tbp]
\caption{\textsc{Search}$(G,\, \optQPx ,\, p,\, \funcNa) \; \rightarrow\; x,\, y)$ \\
Testing $p' \leq p$ to obtain integer $y$.}
\label{algo:test}
\ifextended{}{\small}
\centering
\begin{algorithmic}[1]
\State $p' \gets p$
\While{$p' \leq p$}
\If{$p'$ is considered the first time} 
\State $ \text{feasible}, x, y \gets \optQPx(G,p',\funcNa) $
    \If{ feasible}\label{line:test:before}
        \State $\forall f: y_f \gets \lceil y_f \rceil, \; \lambda'_f {=} \sum_c x_{cf}$
        \State $\Delta \gets \sum_f y_f - p$ \Comment{valid solution for $\Delta {\leq}0$}
        \If{$\Delta = 0$} \Return $x,\, y$ \Comment{direct hit} 
        \EndIf
        \If{$\Delta < 0$}  \Comment{add remaining resources}
            \State $y \gets \Call{Alloc}{p ,\, \lambda' ,\, y }$ 
            \State \Return $x,\, y$ 
        \EndIf
        \If{$\Delta > 0$} \Comment{too many resources} 
            \State $y \gets \Call{DeAlloc}{p ,\, \lambda' ,\, y }$ 
            \If{$ \Call{DeAlloc}{\cdot}$ was feasible}
                \State \Return $x,\, y$ 
            \Else
                \State decrease $p'$ by $\Delta$ \label{line:test:pdelta}
            \EndIf
        \EndIf
    \Else \Comment{infeasible with current $p'$} 
        \State increase $p'$ by 1
    \EndIf
\Else \Comment{$p'$ was considered in previous loop iterations} 
    \State increase $p'$ by 1
\EndIf 
\EndWhile 
\State \Return \Call{Search}{} is infeasible 
\end{algorithmic}
\end{algorithm}

\subsection{Post-processing Surface Results}\label{sec:post-processing-surfaces-results}

As the final building block, the algorithm
\Call{Search}{}~(\Cref{algo:test}) overcomes the inherent drawback of
all surface linearisations: Because linear interpolation along $J$ is
allowed, allocations obtained by \optQPb, \optQPbb, or \optQPc are
 not necessarily integer values.  However, our resources
are not splittable, \eg, number of VM instances, thus allocations have to be
rounded.  Rounding allocations down could result in overutilised
resources and infeasible solutions.  Hence, fractional allocations
have to be rounded up, potentially resulting in allocating more
resources than allowed, violating $\sum_f y_f {\leq} p$.  
Hence, we cannot just round up all allocations. Could it be
possible to round up only some of these allocations and round down
some others? \Call{Search}{} answers this question.  If this
is not possible, another solution is obtained with smaller limit $p'$.  
This way, the algorithm conceptually tries all
$p' {\in} [p_\downarrow,..,p]$ where $p_\downarrow$ is the smallest
number of resources at which the problem is still feasible.

In detail, \Call{Search}{} is a bit more complicated. While it would
be correct to iterate over the sequence $p, p{-}1, \ldots,
p_\downarrow$, this has unacceptable runtime, in particular owing to
the frequent calls to solve \optQPx. But solutions for any
$p'$ and $p'-1$ are most likely to be very similar anyway: If any $p'$ had 
inappropriate assignments, most likely $p'{-}1$'s solution is similarly
inappropriate; so we could skip $p{-}1$ to save runtime. Hence,
\Call{Search}{} modifies $p'$ in larger steps, trying to find a
suitable solution more quickly. The intuition for the step size is the
number of excess allocations due to rounding (see second-last paragraph in this subsection).

\Call{Search}{} invokes \optQPb, \optQPbb, or \optQPc.  The parameter~\optQPx points to the function solving the concrete problem, \eg \optQPb.  Solving \optQPx results in one of three cases: a)~Exactly $p$ resources are allocated; the algorithm is done.  b)~Fewer resources than requested are allocated ($\Delta{<}0$), then \Call{Alloc}{}~(\Cref{algo:alloc}) distributes the remaining $|\Delta|$ resources; the algorithm is done.  c)~More resources than requested are allocated ($\Delta{>}0$), then \Call{DeAlloc}{} removes $\Delta$ resources; \Call{DeAlloc}{}\ifextended{~(\Cref{algo:dealloc})} is basically a twin of \Call{Alloc}{} removing resources one by one; \ifextended{details on
  \Call{Dealloc}{} are provided later}{this paper's
  extended version provides details on \Call{DeAlloc}{}}.
Often, removing all $\Delta$~resources will not be feasible because the current assignment
distributes the demand improperly among the facilities.
Then, only adjusting the assignments itself can help.
This is done by reinvoking \optQPx with a smaller $p'$.
The resulting solution is processed as the original limit $p$ in one of the three cases.
The basic idea is to reduce $p'$ and find a good assignment for which the allocation can be adjusted to use $p$ resources.
At some point $p'$ is so small 
that \optQPx becomes infeasible.

The algorithm \Call{Search}{} finishes in two cases: 
i)~The allocation matches (a) or could be fixed (b,c). 
ii)~After considering all $p'{\in}[p_\downarrow,..,p]$ without success, the algorithm stops without solution.
Case (ii) occurs if \optQPx with $p_\downarrow{+}1$ has a solution whose over-allocation could not be fixed by \Call{DeAlloc}{} due to the assignments while \optQPx with $p_\downarrow$ is infeasible.
In our evaluations this case occurred more frequently with inputs having a very high system utilisation, $\nicefrac{\sum_f \lambda_f }{\sum_f k_f\mu_f} \geq 0.96$ -- in such cases, the minimal necessary allocation matches the limit.

To realise bigger jumps between values for $p'$, \Call{Search}{} examines $p,..,p_\downarrow$ as follows ($p_\downarrow$ is initially unknown):
It starts with $p$ and jumps down to $p' {\gets} p {-} \Delta$
(\Lineref{line:test:pdelta}), skipping as many potential resource limits as resources were over-allocated. 
The idea is that for large over-allocations, a larger step ($\Delta$) is
necessary than for slight over-allocations to change
the assignment. 
It continues jumping down, $p' {\gets} p' {-} \Delta$, until the new $p'$ is infeasible and then tests increasing $p'$, $p' {\gets} p' {+} 1$.
The next values for $p'$ might also be infeasible and increasing $p'$ is continued until a solution for $p'$ can be obtained, $p{=}p_\downarrow'$.
\Call{Search}{} continues increasing $p'$ and additionally skipping already considered $p'$ until $p'{=}p$ is reached and the search terminates without success.
The first stage where $p'$ jumps down allows to quickly find a small $p'$ for which the assignment is appropriate.
If not, the second stage ensures that all $p'\in [p_\downarrow',..,p]$ are tried.

In our evaluations, only for inputs with very high system utilisation, testing many $p'$ limits results in long solving time. 
The remaining inputs needed only $3$--$7$ iterations and separate solving attempts of \optQPx($p'$).
This raises the question if the runtime of these multiple solving times can still compete with solving times of a single solving attempt for, \eg, \optQPa?
An answer is provided by our evaluation in \Cref{sec:appro-compar}.

\begin{inextended}

\begin{algorithm}[tbp]
\caption{\textsc{DeAlloc}$(p ,\, \lambda' ,\, y {=} k_f) \; \rightarrow \; y, T $ \\
    Optimally reduce over-allocation, specified in vector~$y$.\label{algo:dealloc}} 
\centering
\ifextended{}{\small}
\begin{algorithmic}[1]
\Statex $\quad$ \textbf{requires} $\lambda'_f {<} k_f \mu_f, \sum_f \lceil \nicefrac{\lambda'_f}{\mu_f}\rceil {\leq} p$
\Statex $\quad$ \textbf{ensures} $\lambda'_f {<} y_f\mu_f, \sum_f y_f {=} p$
\Statex \hspace{4.9em} minimises $\sum_f \funcN(\nicefrac {\lambda'_f}{\mu_f}, y_f)$
\State $\forall f {\in} F': a_f {\gets} \nicefrac {\lambda'_f}{\mu_f},\; y_f^\text{min} {\gets} \max \{ y_f, \, \lceil a_f \rceil \} $ \label{algo:dealloc:minalloc}
\State $n \gets \sum_f y_f {-} p$
\If{$n {>} 0$}
    \State $y^\text{sub} \gets \Call{MaxCostDrop}{n,\, c_* ,\, y^{\text{max}} }$ \label{line:dealloc:cost}
    \Statex \hspace{4.5em} $\forall f: c_f(y') := \mathfrak{M} {-} \funcN(a_f, y_f{+}y')$
    \Statex \hspace{5em} {\small with $\forall y{=}\funcN(.,1): \mathfrak{M} >y$}
    \Statex \hspace{4.5em} $\forall f: y_f^\text{max} \gets k_f {-} y_f^\text{min}$
    \State $\forall f : y_f =  k_f - y_f^\text{sub}$ \label{algo:dealloc:costs}
\EndIf
\State \Return $y ,\, \sum_f \funcN( a_f, y_f )$
\end{algorithmic}
\end{algorithm}

\begin{lemma}
\label{theo:dealloc}
\Call{DeAlloc}{$p,\, \lambda', k$}~(\Cref{algo:dealloc}) removes $p$ resources at fully allocated facilities $f{\in}F'$ so that $\sum_f \funcN(\nicefrac {\lambda'_f}{\mu_f}, y_f){=}T$ is minimised while ensuring $\lambda'_f {<} y_f\mu_f, \sum_f y_f {=} p$.
\end{lemma}
\begin{proof}
Similar to \Call{Alloc}{}, \Call{DeAlloc}{} uses \Call{MaxCostDrop}{}, but while \Call{Alloc}{} interprets dropped tokens as added resources, \Call{DeAlloc}{}  interprets dropped tokens as removed resources; this make adjustements and a new minimality deduction necessary.
With $n{=}0$ no resources need to be removed $\sum_f k_f {=} p$ having one solution with $T {=} \sum_f \funcN(a_f, k_f)$ minimal; done.
With $n{>}0$, algorithm \Call{MaxCostDrop}{} computes how many resources are removed at which facility.
After removing the resources, $T$ increases by $\sum_f \funcN(a_f, k_f) {-} \funcN(a_f, k_f{-}y_f) = T_\Delta$; so minimising $T_\Delta$ will also minimise $T$. 
As \Call{MaxCostDrop}{} maximises the token costs~(\Cref{theo:maxcostdrop}), we need to flip them around (line~\eqref{line:dealloc:cost}).
This way, $T_\Delta$ is minimal and so is $T$; done.
\end{proof}
\end{inextended}

\begin{inextended}%
\subsection{Linearise Erlang-C Formula} \label{sec:lin-erlangc}

This section describes our approach to linearise the Erlang-C function~\funcEC~\eqref{eq:erlc}.
Our algorithm based on Imamoto\etalcite{Imamoto2008a} obtains basepoints with small linearisation accuracy~\cite{Keller2015}.
This small error is achieved by successively changing the basepoints and estimating their error.
To compute basepoint changes, our algorithm needs the first derivative of the function to be linearised.
This section proposes an efficient, vectorised recursive functions to compute \funcEC, its derivative, and Erlang-C-related functions fast.
This reduces the runtime for obtaining the basepoints.

Beside the efficiency needs, \funcEC becomes not computable with a plain implementation\footnotemark\ for $k{>}145$ and $a$ near $k$.
\footnotetext{Implementation in Python 2.7, Scipy 1.8.0, \SI{64}{bit}-floats; computing $\funcEC(139, 140)$ on our laptop lasts \SI{16}{\milli\second}.}
For such $a$, $k$, the term $a^k$ or $k!$ becomes large enough to overflow\footnotemark\ and further processing is distorted.  
\footnotetext{
The number of floating-point bits determines the largest number which can be encoded, in our case $\approx 1.81 {\cdot} 10^{308}$ with~\SI{64}{\bit}.
If an operation results in a value larger than this limit the result is replaced by a special value, $\texttt{inf}$, indicating the overflow.
}
Alternatively, algebra systems such as Maxima resolves expressions symbolically allowing a greater accuracy paid by a longer runtime and memory usage\footnote{
    Computing $\funcEC(500,1000)$ with Maxima $5$.$26$.$0$ took \SI{0.6}{\second} and consumed \SI{80}{\mebi\byte}.
    $\funcEC(1000,2000)$ took  \SI{3.8}{\second} and consumed \SI{374}{\mebi\byte}.
}.
A better alternative is proposed by Pasternack\etalcite{Pasternack1998}: The recursive function \funcV~\eqref{eq:funcv} replaces the computationally challenging part.
 
For an M/M/k-queuing system with arrival rate $\lambda$, service rate $\mu$, and $k$ servers available, the probability of no jobs $P_{0\,\text{jobs}}$~\eqref{eq:p0} can be rewritten as \eqref{eq:p0withv}, $a{=}\nicefrac{\lambda}{\mu} {<} k$. 
As part of the Erlang-C formula~\eqref{eq:erlc2}, it can be rewritten similarly using \funcV~\eqref{eq:erlc2withv}.
Other queue systems' performance measures use \funcEC.
Examples are the expected number of requests in the system $\funcN(a,k)$~\eqref{eq:funcNwithv} or in queue $\funcN[^Q] (\lambda,\mu,k)$~\eqref{eq:funcNqwithv}, the expected time of request spent in system $\funcTT(\lambda, \mu, k)$~\eqref{eq:funcTwithv} or in queue $\funcTT[^Q] (\lambda, \mu, k)$~\eqref{eq:funcTqwithv}~\cite{Bolch2005}.

\bgroup
\small
\begin{IEEEeqnarray}{ll}
 \funcV(a,k) & = \sum_{i=0}^{k-1} \frac{k!}{i!} a^{i-k} \notag \\
 & = 
 \begin{cases}
 \frac 1 a & \text{ if } k{=}1 \\ 
 \frac k a (\funcV(a,k{-}1){+}1) & \text{ else }  
 \end{cases} \quad \label{eq:funcv} \\
P_{0\,\text{jobs}} & = 
    \left({\frac{k{a^k}}{k!\,\left( k-a\right) }+ \sum_{i=0}^{k-1}\frac{{a}^{i}}{i!}} \right)^{-1} \label{eq:p0} \\
&	= ({\frac{a^k}{k!} \frac{k}{k-a} + \frac{a^k}{k!} 
     \underbrace{\sum_{i=0}^{k-1}\frac{a^i}{i!}\frac{k!}{a^k}}_{=V(a,k)}} )^{-1} \notag \\
&	= \left(\frac{a^k}{k!} \, \frac{k+(k-a)V(a,k)}{k-a}\right)^{-1} \notag \\  
&	= \frac{k!}{a^k}\,\frac{k-a}{k+(k-a)V(a,k)}\label{eq:p0withv} \\
\funcEC(a,\,k) &= \frac{k{a}^{k}}{k!\,\left( k-a\right) } P_{0\,\text{jobs}}
     \label{eq:erlc2}  \\
& 	= \frac{k{a}^{k}}{k!\,\left( k-a\right) } \; \frac{k!}{a^k}\frac{k-a}{k+(k-a) \funcV(a,k)} \notag \\
&	= \frac{k}{k + (k-a) \funcV(a,k)}      \label{eq:erlc2withv}  \\
 \funcN(a, k) & =  a + \frac{a}{k-a} \, \funcEC(a,k) \notag \\
 & = a + \frac{a}{k-a} \, \frac{k}{k + (k-a) \funcV(a,k)} \notag \\
 &= a + \frac{ak}{(k-a)k + (k-a)^2 \funcV(a,k)} \label{eq:funcNwithv} \\
\funcN[^Q](\lambda,\mu,k) &= 
    \frac{\lambda}{\mu k - \lambda} \,\, \frac{k}{k + \frac{\mu k - \lambda}{\mu} V(\frac{\lambda}{\mu},k)} \notag \\
& 	= \frac{\lambda k}{k(\mu k - \lambda) + \frac{(\mu k - \lambda)^2}{\mu} \funcV(\frac{\lambda}{\mu}, k)}
    \label{eq:funcNqwithv} \\
\funcTT(\lambda, \mu, k) & = 
    \frac{1}{\mu} + \frac{k}{(\mu k-\lambda) k + \mu(k-\frac{\lambda}{\mu})^2 \funcV(\frac{\lambda}{\mu},k)} 
    \label{eq:funcTwithv} \qquad \\
\funcTT[^Q] (\lambda, \mu, k) &= \frac{k}{k(\mu k - \lambda) + \frac{(\mu k - \lambda)^2}{\mu} \funcV(\frac{\lambda}{\mu}, k)} \label{eq:funcTqwithv} 
\end{IEEEeqnarray}
\egroup

Two technical performance improvements are worth mentioning:
Pasternack\etalcite{Pasternack1998}'s recursive function is implemented as a loop that not only avoids time needed to push local variables on the stack for each recursive call but also evades the limited stack depth also limiting $k$.
Implementing this improvement results in computation magnitudes faster\footnotemark\ than the other alternatives. 
The function \funcEC with \funcV becomes not computable for $k {\gtrsim} \num{100000}$, sufficient enough for our scenario.
\footnotetext{
Computing $\funcEC(139, 140)$ via \funcV took \SI{0.6}{\milli\second} (Python 2.7).
}
In addition a vectorised version of $\funcV(a,k)$ allows computing several values $\funcV(a_0,...,a_m;\, k_0,...,k_m){=}t_0,...,t_m$ in a row (\Cref{algo:funcv}).
By masking the array fields, different loop iterations depths $k_0,...,k_m$ are processed in one loop~(\Lineref{line:funcv:mask}).
This vectorises our linearisation algorithm and, thus, is much faster than a simple implementation.

\begin{algorithm}[tb]
\caption{Vectorised function $\funcV(a,k)$~\eqref{eq:funcv}}
\label{algo:funcv}
\begin{algorithmic}[1]
\Function{$\funcV(\bar a$, $\bar k)$}{} 
\State $\bar r \gets \nicefrac {1} {\bar a}$ with same length as $\bar a$ and $\bar k$
\For{$i=2..max\{ \bar k \}$}
    \State mask$ \gets \bar k > i$
    \State $r[$mask$] \gets \frac {k[\text{mask}]} {a[\text{mask}]} (r[\text{mask}] + 1 )$ \label{line:funcv:mask}
\EndFor
\Return $r$
\EndFunction
\end{algorithmic}
\vspace{1mm}
{\small \centering With array arithmetic $\bar a {=} a_1,...,a_m:$ element-wise operations $\bar a {\circledast} \bar b \Leftrightarrow \forall i {:} a_i {\circledast} b_i$ and masking $\bar a[1,0,...,0,1] {=} a_1,a_m$.}
\end{algorithm}

Finally, the first derivation of \funcN is needed. 
The basepoints are obtained for parameter~$a$ with fix $k$, $\funcN_k(a)$.
As a result, $\funcN_k(a)$ is differentiated only for $a$; details in~\Cref{sec:apx-funcNdev}.
The first derivative $\frac{d}{da}\funcN_k(a)$~\eqref{eq:derivN} is rewritten with \funcV~\eqref{eq:derivNviaV}.

\end{inextended}

\section{Heuristic}
\label{sec:heuristic}

This section discusses an adaption of the most related (\Cref{sec:rw}) heuristic proposed by Aboolian\etalcite{Aboolian2009}.
While their work also combines the Facility Location Problem with M/M/k-queuing systems at each facility, their problem \optPABD differs from \optQP in three points: 
First, \optPABD minimises the \emph{maximal} response time whereas \optQP minimises the \emph{average} response time. 
Second, assignments in \optPABD are predefined whereas \optQP also decides the assignments.
Third, \optPABD has no resource limit per facility, which \optQP has.
These three differences necessitate adjustments to \optPABD's heuristic.

The resulting heuristic \heu consists of four major parts.
\Call{Alloc}{}~(\Cref{algo:alloc}) computes the optimal allocation for a given assignment to a subset of facilities $F_s{\subseteq}F$.
\Call{Solve}{}~(\Cref{algo:solve}) first assigns requests to the closest facilities in a given facility subset $F_s$ and computes the corresponding allocation with \Call{Alloc}{}.
\Call{Solve}{} is used by \Call{Descent}{}~(\Cref{algo:descent}), which iteratively varies the facility subset to find better subsets.
These variations are limited by two additional facility subsets $F^\text{I}$ and $F^\text{D}$, so that only a local minimum is found.
To find new local minima, \Call{Genetic}{}~(\Cref{algo:genetic}) randomly combines already found solutions.

\begin{algorithm}[tb]
\centering
\caption{Combines local solutions to find new ones.}
\label{algo:genetic}
\ifextended{}{\small}
\begin{algorithmic}[1]
\Function{Genetic}{G ,\, p} 
\State {\small \CommentLine{Create a random population of solutions}}
\State $P \gets \emptyset$, $\, l \gets \lfloor \sqrt{|F|} \rfloor$ \label{line:genetic:initb}
\While{ not enough solutions in $P$ } 
    \State $F_s \gets$ $l$ random facilities from $F$
    \State $F_s, y_f, t \gets \Call{Descent}{G ,\, F_s ,\, p ,\, F ,\, \emptyset}$ \label{line:genetic:initdescent}
    \If {solution not found} increase $l$
    \Else 
        \If {$F_s{\not\in}P$ } add $(F_s,y_f,t)$ to $P$ \EndIf 
    \EndIf
\EndWhile \label{line:genetic:inite}
\While{not enough merge steps are done} \label{line:genetic:mergeb}
    \State \CommentLine{merge two solutions}
    \State $F_s, F'_s \gets$ two random $F_s \in P$
    \State $F_\text{U} \gets F_s {\cup} F'_s$; $F_\text{I} \gets F_s {\cap} F'_s$; 
    \State $F_\text{M} \gets $ three random $f \in F \setminus F_\text{U}$ \Comment{Mutation}
    \State $F_\text{D} \gets (F_\text{U} \setminus F_\text{I}) \cup F_\text{M}$
    \State $F_\text{N} \gets F_\text{I}$ with one $f{\in}F_\text{D}$ added \Comment{Mutation}
    \State $F_s, y_f, t \gets \Call{Descent}{G , F_\text{N} , p , F_\text{D}, F_\text{I}}$
    \If{solution found $\wedge$ $F_s {\not\in} P$ $\wedge$ $t < $ largest $t$ in P}
        \State replace worst solution in $P$ with current
    \EndIf
\EndWhile  \label{line:genetic:mergee}
\State \Return $F_s, y_f, t$ from $P$ with smallest $t$
\EndFunction
\end{algorithmic}
\end{algorithm}
The meta-heuristic \Call{Genetic}{} maintains a finite set of currently best solutions $P$.
At first, an initial set of solutions is randomly generated (\Lineref{line:genetic:initb}--\Lineref{line:genetic:inite}).
This is influenced by two factors: 
The number of maintained solutions, $|P|$, is predefined; maintaining many solutions increases the chance to find different maxima but also increases runtime.
The size of facility subsets $l$ starts from $\sqrt{|F|}$ as in the
original heuristic; if no solution is found with these facilities, $I$
is increased. 
An initial solution is generated by invoking \Call{Descent}{}
(\Lineref{line:genetic:initdescent}) with $l$ randomly chosen  facilities.
\Call{Descent}{} returns a new subset $F_s$, potentially different
from the chosen facilities, the allocation $y_f$, and corresponding
average response time $t$ as the measure of  solution quality.

In \Call{Genetic}{}'s major part
(\Lineref{line:genetic:mergeb}--\Lineref{line:genetic:mergee}), two
random solutions from $P$ are combined to find  new facility subsets. 
Three such new subsets are computed: 
The \emph{intersection} $F_\text{I}$ of both solutions' subsets $F_s, F'_s$ is part of the new offspring facility set $F_\text{N}$.
The \emph{domain} $F_\text{D}$ is the union of $F_s, F'_s$ with three random
new facilities $F_\text{M} $; $F_\text{D} $; it limits the following explorations by \Call{Descent}{}.
The \emph{offspring} $ F_\text{N}$ construction  is the same as in the original heuristic. 
A local solution found by \Call{Descent}{} replaces the worst solution in $P$ if the newly found solution is better than the replaced one.
After a finite number of merge operations, the best solution found so
far is returned.

\begin{algorithm}[tb]
\centering
\caption{Refines solution towards local optimum}
\label{algo:descent}
\ifextended{}{\small}
\begin{algorithmic}[1]
\Function{Descent}{$G ,\, F_s ,\, p ,\, F_\text{D} ,\, F_\text{I} $} 
\State $F_s^\text{min}, y_f^\text{min}, t^\text{min} \gets F_s, \Call{Solve}{F'_s ,\, p}$
\While{smaller $t^\text{min}$ was found}
    \For{$F'_s \in \Call{Neigh}{F_s^\text{min} ,\, F_\text{D} ,\, F_\text{I}}$}
        \State $y'_f,\, t' \gets $\Call{Solve}{$G ,\, F'_s ,\, p$}
        \If{$t'<t^\text{min}$}
            \State $F_s^\text{min}, y_f^\text{min}, t^\text{min} \gets F'_s, y'_f, t'$
        \EndIf
    \EndFor
\EndWhile
\State \Return $F_s^\text{min}, y_f^\text{min}, t^\text{min}$
\EndFunction
\end{algorithmic}
\end{algorithm}
\Call{Descent}{} starts with a given facility subset $F_s$ and searches so-called \emph{neighbouring} subsets for better solutions. 
Aboolian\etal define $F_s$'s neighbourhood (\Call{Neigh}{}) as a set of facility subsets constructed from $F_s$ by a) removing one facility , b) adding one facility, or c) doing both \eqref{eq:descentneigh}.%
\bgroup
\small
\begin{IEEEeqnarray}{lls}
\Call{Neigh}{F_s ,F_\text{D} , F_\text{I}} \; 
& = \left\lbrace F_s {\cup} \{f\} \;\vert\; \forall f {\in} F_\text{D} \setminus F_s \right\rbrace  \qquad & (a) \notag \\
& \cap \left\lbrace F_s {\setminus} \{f\} \;\vert\; \forall f {\in} F_s \setminus F_\text{I} \right\rbrace  \qquad & (b) \notag \\
\IEEEeqnarraymulticol{2}{l}{ \hspace{3em} \cap \left\lbrace F_s {\cup} \{f\} {\setminus} \{f'\} \;\vert\;  \forall f {\in} F_\text{D} \setminus F_s, \,   \forall f' {\in} F_s \setminus F_\text{I} \right\rbrace}  \qquad & (c) \notag \\
\IEEEeqnarraymulticol{3}{r}{\text{ with } F_\text{I} \subseteq F_s \subseteq  F_\text{D} \; \subseteq F \hspace{1em}}  \label{eq:descentneigh}
\end{IEEEeqnarray}
\egroup

If at least one better solution is found in the \emph{subset neighbourhood}, the search continues for the best found solution. The algorithm hence 
descends towards a local minimum.

The facility sets $F_\text{D}, F_\text{I}$ restrict the cardinality of \Call{Neigh}{} and thus the number of instances solved by \Call{Solve}{}. 
This is done by ensuring that $F_\text{I}$ -- the intersection of the two parents -- is always part of $F_s$ and that $F_\text{D}$ -- the modified union of the two parents -- is the subset of $F$ to which $F_s$ can grow.

\begin{algorithm}[tb]
\centering
\caption{Approximate assignment and optimal allocation}
\label{algo:solve}
\ifextended{}{\small}
\begin{algorithmic}[1]
\Function{Solve}{$G, F_s ,\, p$} 
\State $\forall c,f {\in} C {\times} F : x_{cf} \gets 0$  \label{line:solve:begin}
\For{$c, f {\in} C{\times}F_s$ sorted by increasing $l_{cf}$}
    \State dem $\gets \lambda_c {-} \sum_{f'} x_{cf'}$
    \State cap $\gets k_f\mu_f{-}\sum_{c'} x_{c'\hspace{-0.5mm}f}$
    \If {dem $ > 0 \, \wedge $ cap $ >0$}
        \State $x_{cf} \gets \min\{ $ dem $,\,$ cap $  \}$
    \EndIf \label{line:solve:end}
\EndFor
\State \Return $\Call{Alloc}{G, \, \forall f: \sum_c x_{cf}, \, p}$
\EndFunction
\end{algorithmic}
\end{algorithm}
\Call{Solve}{} (\Cref{algo:solve}) computes the assignment and allocation for a given facility subset $F_s$ using $p$ resources.
The first part assigns the demand $\lambda_c$ to the closest facilities.
This assignment is more complex than \optPABD's assignments which 
neither considers resource capacities ($\mu_f$) nor limits ($k_f$).
When a facility's resources are exhausted but demand still exists, this demand is served by another facility.
What was a simple binary assignment to the nearest facility in Aboolian's problem (\optPABD)  becomes an optimisation problem to minimise the total assignment distances (i.e., $\min_{x} l_{cf} x_{cf}$) under capacity and serve-all-demand constraints.
To solve this problem efficiently, the assignment is computed by the greedy heuristic in \Lineref{line:solve:begin}--\ref{line:solve:end}:
Demand is served by the nearest facility $f$ with remaining capacity.
To do so, $c,f$ pairs with the shortest distance are handled first.
If demand remains to be distributed (i.e., $\lambda_c \sum_{f'} x_{cf'} {>}0$) and facility $f$ has enough capacity left ($\mu_f k_f {-} \sum_{c'} x_{c'\hspace{-.5mm}f} >0$) then as much demand as possible ($\Delta$) is assigned from $c$ to $f$.
This way, the demand is moved to free facilities in a way similar to a multi-source breath-first search.

\begin{notinextended}
While this heuristic is in most cases sufficient, in rare cases swapping the assignments $x_{cf}$ would further improve the solution. This paper's extended version~\cite{Keller2015b} discusses these cases and proposes an extension to \Call{Solve}{}.
\end{notinextended}

\begin{inextended}
While this heuristic is in most cases sufficient, it does not compute the total average RTT in all cases.
Let us consider two clients $c,c'$ with demand, $\lambda{=}1$, and facilities $f,f'$ with one capacity, $\mu{=}1$, having the following latencies $l_{cf}{=}1,\,l_{cf'}{=}l_{c'\hspace{-0.5mm}f}{=}2,\,l_{c'\hspace{-0.5mm}f'}{=}4$. 
The assignment heuristic (\Lineref{line:solve:begin}--\ref{line:solve:end}) computes $x_{cf}{=}x_{c'\hspace{-0.5mm}f'}{=}1$. 
First the $c,f$ pair is visited, then $c'\hspace{-0.5mm}f$ is not possible as capacity is exceeded, and finally $c'\hspace{-0.5mm}f'$ is paired. 
These assignments result in total RTT $5$ but the optimum is $4$ with assignment $x_{cf}{=}x_{c'\hspace{-0.5mm}f'}{=}1$.

One way to handle this case, is to check for total distance reduction when swapping assignments in a post-processing. 
The swapping procedure is very time consuming as considering all pairs only once is not sufficient to obtain the optimal solution.
The described special case rarely occurred in our test instances and we had favoured short runtime, so that we omitted such a post-processing.

The presented assignment heuristic applies an extended nearest-facility assignment rule and, hence, is more restrictive than \optQP.
We tested an alternative implementation for \Call{Solve}{}: Assignment and allocation are obtained by our fastest optimisation problem \optQPc for facility subset $F_s$.
However, \Call{Solve}{} as the basic function of \heu is called very often and the runtimes were unsatisfactory longer than the presented implementation, much longer than any other presented approach to solve \optQP discussed in \Cref{sec:eval}. 
Because of this, we entirely skipped this variant. 
\end{inextended}

%% file: eval.tex
\section{Evaluation} \label{sec:eval}

\begin{inextended}
\subsection{Obtaining the Basepoints} \label{sec:eval:bps}

The linearised formulations' quality and solving time depend on the choice of basepoints.
For all optimization problems above, more basepoints improve the linearisation accuracy and improve solution quality.
But more basepoints also increase the search space  and extend solving time.

Four control factors influences the accuracy: 
a)~The number of basepoints~$m$ for one curve $\funcNa_j(a)$; 
b)~the number of curves/allocated resources $n$, $1{<}n{\leq}k_f$; 
c)~the linearisation interval's upper bounds for $a$ and $j$;
d)~the basepoint positions themselves.

\begin{table}
\centering
\caption{Investigated values for $J$ and $m$.} 
\def\arraystretch{1.2}
\label{tbl:basepointcombinations}
\begin{tabular}{rll|l}
\toprule
{$J$'s \footnotesize shorthand} & $J = n$ & $n$ & $m$\\
\midrule
k100 & $1,...,100$ & $100$ & $30$  \\
fib & {\small$1,2,3,5,8,13,21,34,55,89,100$} & $11$ & $15$\\
$2^i$ & $1,2,4,8,16,32,64,100$ & $8$ & $8$ \\
$3^i$ & $1,3,9,27,81,100$ & $6$ & $6$ \\
$4^i$ & $1,4,16,64,100$ & $5$ & $4$\\
k3 & $1,50,100$ & $3$ \\
\bottomrule
\end{tabular}
\end{table}

For factor~(c), we allow $0{<}a{<}0.98j$ similar to our previous work~\cite{Keller2015} and $1{<}j{\leq}100$. 
Setting the maximal number of available resources $\forall f: k_f{=}100$ is large enough for our evaluation. 
Our findings can be applied for larger values of $k_f$ that decrease either the linearisation accuracy (for fixed $m$) or the number of basepoints $m$ (for fixed  accuracy)~\cite{Keller2015}.
For factor~(c), our algorithm~\cite{Keller2015} is applied obtaining basepoint sets with high linearisation accuracy.
For factors~(a) and~(b), different configurations for $m$ and $J$ are investigated.
\Cref{tbl:basepointcombinations} shows these configurations.
For example, $J{=}[1,50,100]$ means $n=3$ curves with $y{=}1$, $5$, or $100$ are considered, indicated by the shorthand k3.
For each list $J$, the corresponding $\funcNa_j(a)$, $j{\in}J$ are approximated by $m=4$, $6$, $8$, $15$, or $30$ basepoints.
The table describes $30$ different configurations~$(J,m)$.

The concrete values for $m$ result from the following thoughts: 
When adding a basepoint to a PWL function (for one curve), accuracy improvement has diminishing returns~\cite{Keller2015}, \eg, adding a basepoint to 4 basepoints helps more than adding one to 8 basepoints. 
So, many low values for $m$ cover the interval where accuracy changes significantly and a few large values  provide hints towards  asymptotic behaviour.
The cases for $m=2$ or $3$ basepoints are dropped; $m{=}2$ results in one line segment and with $m{=}3$, having two line segments, approximating our convex function still results in low accuracy.

The concrete values for $J$ result from the following thoughts:
As discussed in \Cref{sec:lin:curvethin}, dropping larger $j{\in}J$ from $J$ results in less inaccuracy than dropping small $j$, so we choose sequences like $a_n{=}2^n$ favouring many small values over few large values.
Additionally, $J{=}k100$ is added as a baseline comparison using all curves.

\begin{figure}[tb]
\centering
\includegraphics[scale=0.95]{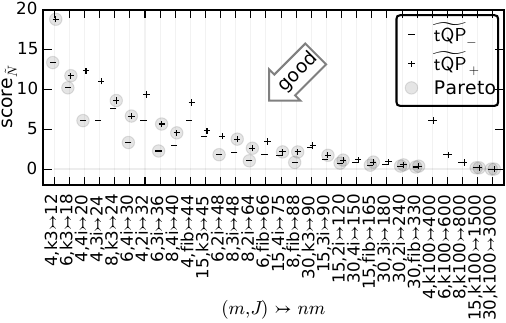}
\caption{
Surface linearisation accuracy $\epsilon_{\tilde N}$~\eqref{eq:pwlerror}.
}
\label{fig:errorscore}
\end{figure}

\Cref{fig:errorscore} shows the error $\epsilon_{\tilde N}$ defined in equation~\eqref{eq:pwlerror} for each $(m, J)$~configuration~(\Cref{tbl:var}) plotted as a function of the number of basepoints $nm$; its labels encode the different configurations for $m, J$ (using the above shorthand for $J$).
Two symbols correspond to the two triangle orientations~(\Cref{fig:surfacemodel}) used to define the approximating surface~$\tilde S$.
\bgroup
\small
\begin{IEEEeqnarray}{l}
\epsilon_{\tilde N} := max_{a,k} \vert\, \tilde N(a,k) - N(a,k) \,\vert \notag \\
\tilde N \text{ is paremeterised by } (m,\, J,\, \text{triangle orientation} + / - ) \qquad \label{eq:pwlerror}
\end{IEEEeqnarray}
\egroup

\noindent \Cref{fig:errorscore} also illustrates the trade-off between two conflicting metrics: Few basepoints $nm$, shown on the left, and low error $\epsilon_{\tilde N}$, shown on the right.  Pareto-optimal solutions (for each triangle orientation) are shown grey-shaded.

For our evaluation we choose the basepoint sets as follows: 
Since many basepoints $nm$ increase the search space size, large $nm{>}330$ were skipped; the lowest error of these skipped configurations is only marginally smaller than the lowest error of the remaining configurations.
The remaining $nm$ values were split in three partitions.
In each partition, one configuration was chosen with the lowest error $\epsilon_{\tilde N}$ satisfying the Pareto property: $(m,J) {=} (15,2^i), (8,3^i), (6,4^i)$.
In addition, configuration $(8,$k100$)$ from the baseline basepoint sets (k100) was added. 
The number of basepoints $m{=}8$ provides a good trade-off between low error and few basepoints, \eg, with $25\%$ fewer basepoints the error nearly doubles $(6,$k100$)$ and with $88\%$ more basepoints the error is only divided by four $(15,$k100$)$.

\end{inextended}

\subsection{Approach Comparison}
\label{sec:appro-compar}
The main question is which presented approach (one of optimisation formulations  or
heuristic) solves \optQP\ best, where \emph{best} is described by two
conflicting metrics: Quality (accuracy) and solving time.  We consider 
examples in which we vary four factors: 
Topology \hG, demand distribution \hD,
basepoint set \hB,  and resource limit \hp.

\begin{inextended}

\begin{table*}
\centering
\caption{Topology overview.} \vspace*{-2mm}
\def\arraystretch{1.3}
\small
\label{tbl:topologies}
\begin{tabular}{l@{\ } l S@{} S@{\hspace{3mm}} S@{}  S@{} l@{}} 
  \toprule 
   {Name$_\text{source}$} & { Mark} & {Size} & { Density} & { $\varnothing $RTT} & { $\text{diam}_\text{RTT}$} & { Degree Perc.} \\ 
   \multicolumn{1}{c}{\centering\hG} & & { $|N|$} & {$2 \nicefrac{|N|}{|E|(|E|-1)}$} & {[ms]} & {[ms]} & [ min $25\,$\% median $75\,$\% max ] \\
  \midrule 
  Columbus {\tiny zoo} &  d & 31 & 0.211 & 326.5 & 1203.8 & [   1.0   3.0   4.0  11.0  12.0 ] \\ 
  Cesnet20 {\tiny zoo} & n- l- & 38 & 0.064 & 189.5 & 387.3 & [   1.0   1.0   1.0   3.0  15.0 ] \\ 
  giul39 {\tiny snd} &  d & 39 & 0.116 & 447.9 & 852.1 & [   3.0   3.0   4.0   5.0   8.0 ] \\ 
  pioro40 {\tiny snd} & n- l+ & 40 & 0.114 & 545.6 & 1140.8 & [   4.0   4.0   4.0   5.0   5.0 ] \\ 
  Colt {\tiny zoo} & n/ l/ & 149 & 0.016 & 546.8 & 1338.5 & [   1.0   1.0   2.0   3.0  18.0 ] \\ 
  UsCarrie {\tiny zoo} & n/ l/ & 152 & 0.016 & 732.7 & 2076.6 & [   1.0   2.0   2.0   3.0   6.0 ] \\ 
  Cogentco {\tiny zoo} & n/ l/ & 186 & 0.014 & 865.9 & 2260.1 & [   1.0   2.0   2.0   3.0   9.0 ] \\ 
  Kdl {\tiny zoo} & n+ l+ & 726 & 0.003 & 1476.4 & 4317.1 & [   1.0   2.0   2.0   3.0  10.0 ] \\ 
  intercon* {\tiny zoo} & n+ l/ & 1108 & 0.002 & 595.8 & 1758.4 & [   1.0   1.0   2.0   3.0  54.0 ] \\ 
  kingtrace* & n+ l- & 1819 & 0.180 & 47.9 & 726.4 & [   3.0 220.0 242.0 467.0 559.0 ] \\ 
  \bottomrule
  \multicolumn{7}{l}{(*) Marked topologies were too large to be solved efficiently.}
\end{tabular}
\end{table*}

\end{inextended}

Candidate topologies were collected from different sources: sndlib\footnote{\label{fn:lat} Round trip times were approximated by geographical distances~\cite{Kaune2009}. Nodes without geolocation positions were removed and their neighbours were directly connected.}~\cite{Orlowski2007}, 
topology zoo\cref{fn:lat}~\cite{Knight2011}, and kingtrace\footnotemark~\cite{Gummadi2002}.
\footnotetext{In kingtrace, a sparse matrix specifies point-to-point latencies. Some latencies were only available in one direction. We assume the same latency for the opposite direction; otherwise those nodes would have to be discarded.}
The topology sources offer $534$ candidate topologies from which $8$ were selected \ifextended{ (\Cref{tbl:topologies})}
by three properties: 
First, the number of nodes increases the problem size quadratically. We
selected three small\ifextended{~(marked as~n-)} (with at least 20 nodes), three medium\ifextended{~(n/)}, and one\footnotemark\ large\ifextended{~(n+)} topologies.
\footnotetext{Only one of three large topologies were solved within the time limit, see extended version~\cite{Keller2015b}.}
Second, the round trip times contribute to the objective function.
The selected topologies have low\ifextended{~(l-)}, medium\ifextended{~(l/)}, or high\ifextended{~(l+)} average round trip times, $\varnothing\text{RTT} {=} \nicefrac{1}{|E|} \sum_{vv'} l_{vv'} $.
Third, resource distributions will be degree-dependent, \eg, few resources on poorly connected nodes and many resources on well connected nodes.
Two topologies were explicitly selected \ifextended{(marked by~d)}
whose  quartiles of node degrees differed from the other
topologies\ifextended{~(\Cref{tbl:topologies})}.

The available resources, in total $\sum_f k_f {=} 5|N|$, are assigned
to nodes weighted by degree.
All resources are homogeneous with $\hMu {=} \SI{100}{\req\per\second} {=} \mu_f,\, \forall f$.

The second factor \hD describes how the individual Poisson processes' arrival rates $\lambda_c$ are assigned to nodes $c$. For each evaluation run, we choose all $\lambda_c$ randomly but keep them fixed within one run. Averaged over all runs, the expected rate is $\hat \lambda$ and all $\lambda_c$ are distributed according to one of three different distributions:
a) Gaussian normal distribution with small standard deviation, $\hD {\hat =} \text{N}(\hat\lambda, \nicefrac {\hat\lambda} {20} ) {=} \text{N}_1$ %
b) Gaussian normal distribution with large variance  $\hD {\hat =} \text{N}(\hat\lambda, \hat\lambda ) {=} \text{N}_2$, %
c) and exponential distribution $\hD {\hat =} \text{Exp}(\hat\lambda)$ with even stronger variations to reflect local hot spots.  %
We ignore nodes with negative arrival rates, hence the random variables $\lambda_c$ are defined as follows: $\forall c {:} \lambda_c {=} \max \{ 0, X\}$ with $X {\sim} \hD$.\footnote{This slightly skews the expected values, but for the concrete scenarios, this effect was negligible.}

The mean arrival rate is $\hat \lambda {=} \nicefrac{0.98}{2} \sum_f k_f \mu_f$.
This way, on average half the available resources are needed to handle the
demand. This is the typical scenario where taking both round trip time and queuing delay pays off -- if almost all resources are needed anyway, there is no freedom of choice; if few resources are needed, queueing delay is unimportant and the problem 
degenerates into a simple RTT optimization problem.  Technically, the reduced factor $0.98$ instead of $1$ reflects the linearisation bound $\alpha_m{=}0.98k$\ifextended{~(factor (d) in \Cref{sec:eval:bps})} and ensures that all queuing systems are in steady state.

The third factor is the resource limit $\hp{=}\sum_fy_f$.  It influences the average resource utilisation and, hence, the average queuing delay.  The higher the resource limit, the lower the resource utilisation.
Since the total number of available resources ($\sum_f k_f$) differs
among the topology sizes\ifextended{~(\Cref{tbl:topologies})}, the
resource limit is cast as a function of this total number of
resources, $\hp {=} \lceil a \, \sum_f k_f \rceil$; for values
$a{<}0.5$ the input becomes infeasible due to the selection of $\hat
\lambda$, where half of the resources are needed to handle the demand.
For the evaluation, factor $\hp$ uses $a{=}\numlist{0.5625; 0.625 ;
  0.6875; 0.75 ; 0.8125; 0.875 ;0.9375}$; where $a=0.56$ allows
slightly more resources than needed for the demand and $a=0.94$ allows using nearly all resources.

While the first three factors vary the topology, the fourth factor  basepoint set \hB varies the linearised problems.
\begin{inextended}
From \Cref{sec:eval:bps}, \hB is $(m,J) {\in} \lbrace (15,2^i), (8,3^i), (6,4^i), (8,k100) \rbrace$.
\end{inextended}
\begin{notinextended}
Each basepoint set is described by two parameters~$m$ and
$J$~(\Cref{sec:lin:curvethin}), where $m$ is the number of basepoints used for
each PWL function~$\funcNa_j$, $j{\in}J$. 
Numerous combinations are possible but considering all of them is impractical.
Instead,  a preliminary evaluation compared values for $m$ and
reasonable sequences for $J$ for their linearisation accuracy; the extended version~\cite{Keller2015b} describes  this in detail.
Four sets of different sizes with high corresponding accuracy are selected $\hB{=}(m,J) {\in} \lbrace (15,2^i), (8,3^i), (6,4^i), (8,k100) \rbrace$, where $J{=}a^i$ refers to the sequence $J{=}[a^1,a^2,..100]$ and $k100$ refers $J{=}[1,..,100]$.
\end{notinextended}
 
Putting all factors together, $168$ factor combinations~(\hG, \hD,
\hp) are considered and for each one, $50$ random demand realisations are generated. 
This results in \num{8400} different \emph{configurations} to be solved by
either the optimization solver or the heuristic.
The linearised problems ($\optQPa, \optQPb, \optQPbb, \optQPc$)  use $3$
basepoint sets $\hB{=}\lbrace (15,2^i), (8,3^i), (6,4^i) \rbrace$.
In addition, \optQPa is solved with basepoint set $(8,k100)$\ifextended{~(\Cref{sec:eval:bps})} and serves as the baseline comparison; this case considers all curves and, hence, has the highest linearisation accuracy and the best solution quality.
The randomised heuristic solves each configuration up to $15$ times to obtain statistically meaningful results.
But for medium or large topologies the heuristic did not compute
a solution in reasonable time because after $6$ hours it still builds
up its initial population.
As an optimisation solver, Gurobi\footnotemark{} is used 
\footnotetext{Gurobi version 5.6.3, Python version 2.7.9}
and is configured to stop solving after one hour\footnotemark{}.
Especially for larger topologies, this often causes optimality gaps up
to $20\%$. 
\footnotetext{Hardware used: 2 of 6 Xeon X5650 Cores, \SI{2.7}{\giga\hertz}, \SI{4}{\giga\byte} RAM}

\begin{notinextended}
The extended version~\cite{Keller2015b} compares solution quality (average response time) for single factor combinations; this paper  aggregates  the core findings.
\end{notinextended}

\begin{inextended}
The remaining section first discusses example results (\Cref{sec:eval:comp:detail}) and then presents aggregated statistics revealing the core findings.
\subsubsection{Detailed Results} \label{sec:eval:comp:detail}
\begin{figure*}[tbp]
\centering
\begin{subfigure}[t]{0.9\textwidth}
\includegraphics[width=\textwidth]{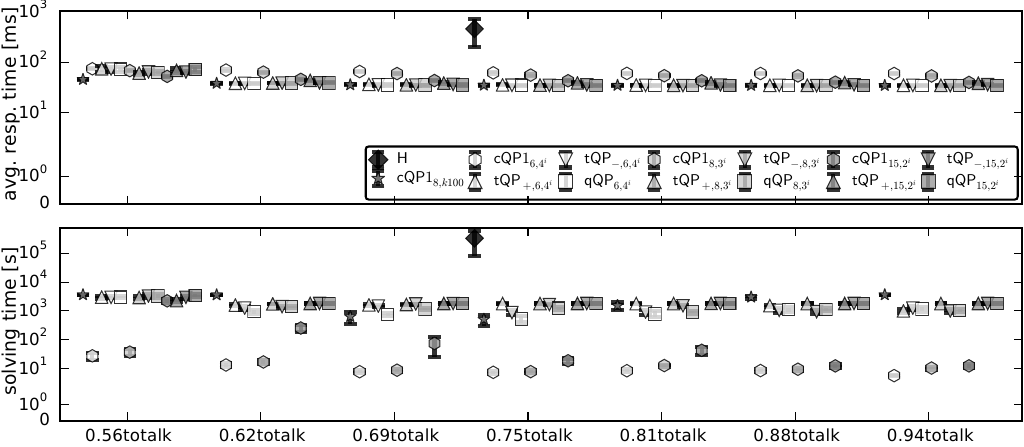}
\caption{For $\hG=\texttt{Colt}$, $\hD=\text{Exp}$}
\label{fig:algo_colt_exp}
\end{subfigure}
\vspace{2mm}
\begin{subfigure}[t]{0.9\textwidth}
\includegraphics[width=\textwidth]{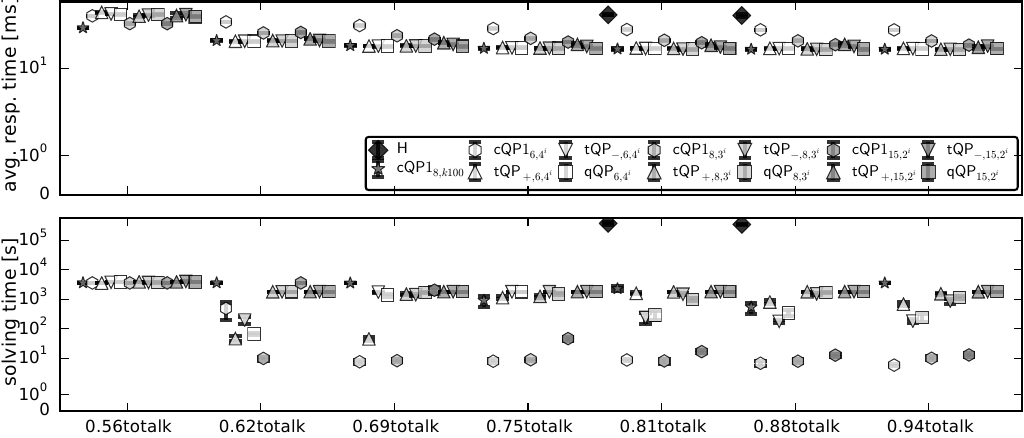}
\caption{For $\hG=\texttt{Colt}$, $\hD=\text{N}_1$}
\label{fig:algo_colt_n}
\end{subfigure}
\vspace{2mm}
\begin{subfigure}[t]{0.9\textwidth}
\includegraphics[width=\textwidth]{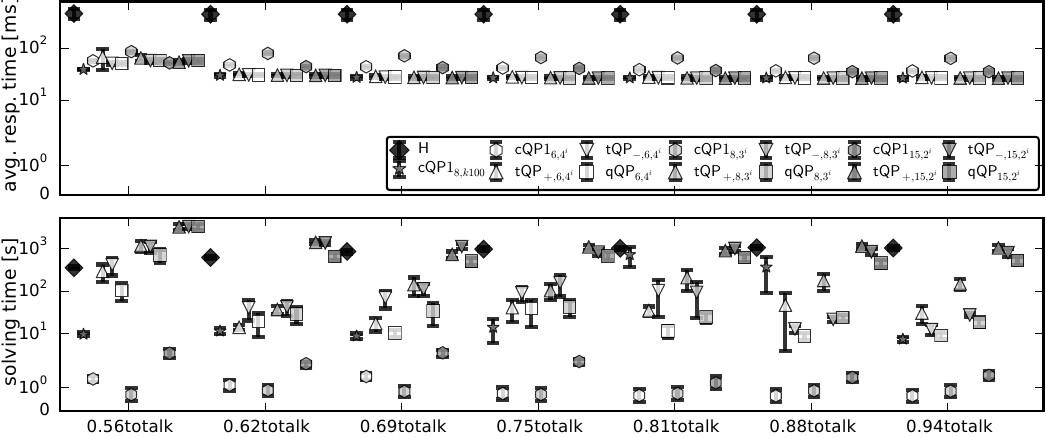}
\caption{For $\hG=\texttt{Pioro}$, $\hD=\text{N}_2$}
\label{fig:algo_pioro_exp}
\end{subfigure}
\caption{Picked detail comparison of algorithms -- quality and runtime.\label{fig:algo_all}}
\end{figure*}

\Cref{fig:algo_colt_exp} shows two plots with either the solution quality as the average response time and the corresponding average solving time with $95\%-$confidence intervals.
They compare all $14$ solving approaches, represented by different symbols, for different utilisation levels $\hp$.
\Cref{fig:algo_colt_exp} shows results for medium topology \texttt{Colt} with exponential demand distribution.
Below, \Cref{fig:algo_colt_n} shows results for the same topology but with normal demand distribution $\text{N}_1$.
At last, \Cref{fig:algo_pioro_exp} shows results for smaller topology \texttt{Pioro} with demand distribution $\text{N}_2$.

In the shown plots, the baseline approach $\optQPa_{8,k100}$ computes the best solution (smallest avg. resp. time) among the other approaches but is the slowest approach, except for small topologies (e.g. \Cref{fig:algo_pioro_exp}) where other approaches are significant slower.
\Cref{sec:eval:comp:baseline} describes how $\optQPa_{8,k100}$ performs among the other topologies.

The different formulations $\optQPa_j$, $\optQPb[+,j]$, $\optQPb[-,j]$, and $\optQPc_{j}$ are grouped together and 
the three basepoint sets $j\in\lbrace(6,4^i)$, $(8,3^i)$, $(15,2^i)\rbrace$ form three groups distinguished in \Cref{fig:algo_all}a-c by different grey levels. 
Among these four approaches, the thinned curve formulation $\optQPa_j$ is the simplest and the fastest in each group, sometimes magnitudes faster, while the quality is a bit worse than the other approaches.
When comparing the surface linearisations, the quadrilateral-based formulation~$\optQPc$ is as good as the other two formulations but is solved faster in some cases.
\Cref{sec:eval:comp:surfaceform} describes how the surface linearisations perform among all configurations.

Finally, the last solving approach uses the genetic algorithm as an heuristic. For medium and large topology, the heuristics's solving time exceed a maximum runtime of six hours in most cases; for those cases no measurements are available.
\Cref{sec:eval:comp:heuristic} describes how the heuristic performs among all configurations.

\end{inextended}

\subsubsection{Baseline Algorithm} \label{sec:eval:comp:baseline}
\begin{figure}[tbp]
\centering
\includegraphics[scale=0.90]{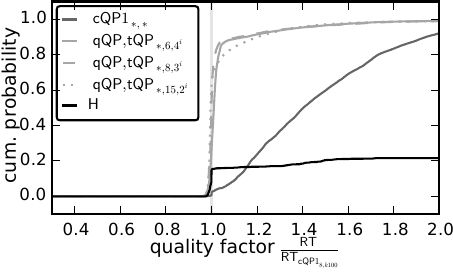}
\caption{Comparing the quality of baseline approach $\optQPa_{8,k100}$ with the other approaches as a CDF plot.}
\label{fig:algo_baseline}
\end{figure}

The approach $\optQPa_{8,k100}$ is the baseline formulation. We
compare solution alternatives to the baseline by computing the ratio
of the alternative's quality to the baseline's quality individually
for each of the \num{8400} configurations
(ratio $>1$ means the baseline algorithm performs better).  As this
produces many individual results, we group similar solution
alternatives together (see below for details). The ratios in these
groups are then jointly described by empirical cumulative density
functions~(ECDFs), \eg \Cref{fig:algo_baseline}.

We identify five groups of similar solution alternatives. 
The first group contains all thinned curves approaches
($\optQPa_{6,4^i}$,$\optQPa_{8,3^i}$,$\optQPa_{15,2^i}$); for this
group, $99.4\%$ of all configurations are solved better by the baseline approach. 
For the remaining configurations, the lowest ratio is $0.97$; the better solutions are caused by \Call{Alloc}{} which uses the exact queuing delay function. 

The second, third, and fourth group contain all surface linearisations
with basepoint sets $\hB=(6,4^i)$, $(8,3^i)$, and $(15,2^i)$. 
For these groups, the configuration are solved similarly well as with the baseline approach.
Surprisingly,  $26\%$,  $51\%$, and $63\%$ of all configurations
result in ratios below~$1$, with the lowest ratios being $0.95$, $0.96$, and $0.96$.
This is caused by algorithm~\Call{Alloc}{} called by the post-processing~\Call{Search}{}.

The fifth group contains all heuristic solutions, which are very good for small topologies but for the other topologies the quality was magnitudes worse. 
\Cref{sec:eval:comp:heuristic} has more details about the heuristic.

To summarize, solutions obtained by \emph{$\optQPa_{8,k100}$ are good references for the expected solution quality} and 
\emph{solutions obtained by \optQPbb[], \optQPc are similarly good}.
Better solution qualities than the baseline algorithm's quality
involved using \Call{Alloc}{}; the difference to the baseline is always small and results from using the exact instead of the linearised queuing delay. 
The difference relates to the linearisation accuracy.

\subsubsection{Thinned Curves \vs Surface Linearisations} \label{sec:eval:comp:surfaceform}

One of the goals of this paper was to compare univariate
vs.\ bivariate linearisations of the time-in-system function. To do so,
we compare here the quality ratios obtained from either thinned curve
linearisation (Section~\ref{sec:lin:curvethin}) or surface
approximations (Sections~\ref{sec:lin:surfacetriangle} and
\ref{sec:lin:surfacequad}). For the comparison, we fix the basepoint
sets $j\in\lbrace(6,4^i)$, $(8,3^i)$, $(15,2^i)\rbrace$. For each
basepoint set, we compute the quality ratio by dividing the solution
quality  of $\optQPb[+,j]$, $\optQPb[-,j]$, or
$\optQPc_{j}$  by that of  $\optQPa_j$. The resulting ratios again
give raise to three ECDFs, one per basepoint set. We do the same thing
for the solving times, obtaining three more ECDFs.

\begin{figure}[tbp]
\centering
\includegraphics[scale=0.90]{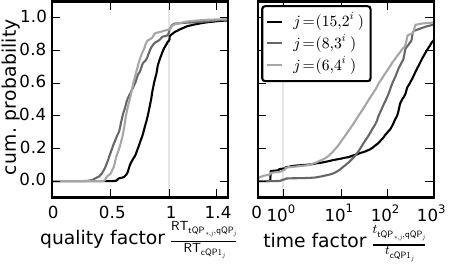}
\caption{Comparing the quality and solving time ratios of $\optQPa_{j}$ with $\optQPb[*,j]$ and $\optQPc_j$.}
\label{fig:algo_surfacegroup}
\end{figure}

\Cref{fig:algo_surfacegroup} shows ECDFs of these quality and time
ratios for each of the basepoint sets.
Most ($92.6\%$, $90\%$, $86.7\%$) quality ratios were smaller than 1. However, solving configurations with surface-base approaches takes magnitudes longer than with $\optQPa_j$; 
$32\%$, $48\%$, and $70\%$ of the surface approaches took $100$
times longer than the thinned curves approach.

This longer solving time has two causes: 
First, $\optQPa_j$ is much simpler and has a fast post-processing~(\Call{Alloc}{}) while \optQPb[] or \optQPc are invoked multiple times by \Call{Search}{}. 
A better but more complex post-processing, adjusting an over-provisioned
solution obtained by \optQPb[] or \optQPc, could reduce their solving
times; this is for further study.
Second, Gurobi stopped improving the solution after one hour; so without this limit, the time factor likely increases.
To conclude, \emph{among the linearised problems, the thinned curve approach $\optQPa$ showed good quality with the shortest solving time in most of the cases}.

\subsubsection{Triangle \vs Qudrilateral Surfaces}
\label{sec:triangle-vs-qudr}

Similar to the comparison of curves vs.\ surfaces, we are interested
in characterizing the behaviour of the different surface approaches. To
this end, we compute quality and solving time ratios of the triangles
divided by the quadrilateral approaches
(\Cref{fig:algo_surfacegroup2}). 

\begin{figure}[tbp]
\centering
\includegraphics[scale=.90]{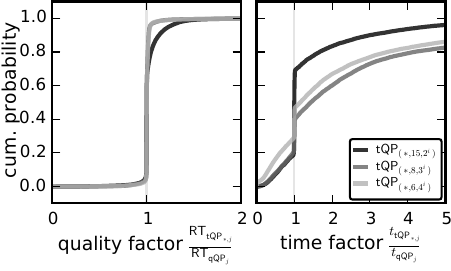}
\caption{Comparing the quality and solving time ratios of $\optQPa_{j}$ with $\optQPb[*,j]$ and $\optQPc_j$.}
\label{fig:algo_surfacegroup2}
\end{figure}
For each group, in $65\%$, $72\%$, and $62\%$ of all configurations,
the quadrilateral approach  \optQPc is faster than the other formulations. 
But in $27\%$, $18\%$, and $17\%$ of all configurations, \optQPc is $10\%$ slower than the other formulations. 
For each group, $87\%$, $87\%$, and $88\%$ of all configurations are solved equally good or better with \optQPc than with the other formulations 
and in the remaining configurations \optQPc is worse than the other
formulations. This confirms the structural arguments for the
similarity between these two approaches. 
The first \Call{Search}{} iteration results in different over-utilised solutions' when using \optQPc or \optQPb[*].
Then, \Call{Search}{} continues solving problems with different limit $p$ and that results in different solutions.
To conclude, \emph{the problem~\optQPc embedded in algorithm
  \Call{Search}{} obtains in most cases a very good solution in less
  time than the triangle-base problems}.
 
In addition, we compared solutions of \optQPb[*] and \optQPc of a subset of $200$~configurations.
Doing so, we could support two arguments present previously:
(a)~The triangle orientation influences the solution quality -- the triangle orientation is the only difference between \optQPb and \optQPbb, where \optQPb's solution quality were always smaller/better than \optQPbb's solution quality.
(b)~Using quadrilaterals instead of triangles is reasonable -- the solution and solution quality for \optQPb and \optQPc matched in all configurations.

\subsubsection{Heuristic} \label{sec:eval:comp:heuristic}

The heuristic has a long solving time\footnotemark, so a six-hour time
limit was imposed; solving times beyond this limit are entirely impracticable for our scenario.
\footnotetext{The heuristic was implemented in Python and utilised NumPy to speed up computations. In contrast, Gurobi  is highly optimised. This biases the solving time comparison a bit, but the results reveal that even a highly optimised heuristic will be outperformed by Gurobi for large topologies.}
Each configuration was solved $15$ times to average out random
variations in  the heuristic \Call{Genetic}{} itself.
Some of these solving attempts were successful while others failed for the same configuration.
\begin{inextended}
\Cref{tbl:heu_succ} shows how many attempts were successful, grouped by topologies.

\begin{table}[tbp]
\centering
\caption{Obtained solutions of all configurations with heuristic.\label{tbl:heu_succ}}
\begin{tabular}{l|cccc} 
\toprule
Name, $G$ & {\small Columbus} & {\small Cesnet20} & {\small guil39} & {\small pioro} \\ 
\midrule
Size, $|N|$ & 31 & 38 & 39 & 40  \\ 
$\#\,$Sol. $\%$ & 30 & 39.7 & 40 & 89.3 \\
\bottomrule 
Name, $G$ & {\small Colt} & {\small UsCarrier} & {\small Cogento} & {\small Kdl} \\ 
\midrule
Size, $|N|$ &  149 & 152 & 186 & 726 \\ 
$\#\,$Sol. $\%$ &  0.2 & 0 & 0.7 & 0 \\ 
\bottomrule 
\end{tabular} 
\end{table}
\end{inextended}

In total, $85.1\%$ of all attempts were not solved in time and most of them do not advance beyond \Call{Genetic}{}'s initial phase.
The size of the topology corresponds directly to a large neighbourhood visited several times in \Call{Decent}{}~(\Cref{sec:heuristic}).
This causes the dramatically long solving time.
\begin{inextended}
Improvements are possible, such as caching a history of \Call{Decent}{} calls avoiding visiting same neighbourhoods again, but a first implementation showed high memory demand of the cache and high runtime due to many cache checks.
\end{inextended}
Concluding, \emph{using the genetic heuristic} (in its current version) \emph{is impractical with its significantly worse quality and solving time}.

\subsection{Scenario Variants}
\label{sec:scenario-varia}

This section investigates how application performance and resource distribution influence the average response times. 
The same setup as in the previous section was used but with different values for the compared factors focusing on the new investigation.

For the application performance, we consider short, medium, and long
request processing times corresponding to high, medium, and low service rates; $\hMu{=}1; 100; \SI{10.000}{\req\per\s}$.
The same total number of resources $\sum_fk_f{=}5|N|$ were geographically distributed in five different ways (factor \hS).
(a)~$\hS{=}\texttt{d5}$: $|N|$ available resources are assigned to the
$5$
nodes with largest degrees\footnotemark; 
\footnotetext{For all degree-based selections, the node ID was the tiebreaker.}
(b)~$\hS{=}\texttt{d}$: resources are distributed
across all nodes, weighted by degree (this was used in our previous work); 
(c)~$\hS{=}\texttt{d}^2$: node degrees are squared,  amplifying the effect of having more resources at better connected nodes.
(d)~$\hS{=}\texttt{c}$: As a baseline case, all available resources are placed at a single node having the lowest total latency to all other nodes,. This baseline case minimises the average queuing delay.
(e)~$\hS{=}x$: All nodes have $100$ available resources, eliminating
effects of capacity limits ($k_f$)  and acting as an upper quality
bound.  

The demand distribution follows only two mathematical distributions, $\hD{\in}\lbrace\text{N}_2, \text{Exp}\rbrace$.
The resource limit \hp is restricted to low and high facility utilisation, $\hp{=}\lceil 5a|N| \rceil $, $a{=}\numlist{0.51;  0.6;  0.75}$.
The same topologies are considered as in the previous section\ifextended{~(\Cref{tbl:topologies})}.
The resulting $720$ factor combinations each have $50$~realisations of demand distribution, resulting in \num{36000} configurations.
These configurations are solved by \optQPc with basepoint set $\hB{=}(6,4^i)$.

At first, the configurations are grouped into combinations of factors $(\hMu, \hp, \hD)$.
Within one group, the resource distribution factors $\hS$ are
compared. 
For all groups, configurations with $100$~resources everywhere ($\hS{=}x$) have, as expected, the shortest round trip times and response times. 
Configurations with resources at a single site ($\hS{=}c$) have, as
expected, the longest round trip times and the shortest queuing delays in all groups.
In addition, for some groups this factor results in the shortest response time showing that queuing delay drop compensates the longest round trip time.
Consequently, \emph{using resources at multiple sites does not always have a shorter response times than a single-site resource allocation}.
Which factor $\hS$ results in the second-shortest response time depends on the topology $\hG$ but was independent of $\hMu, \hp, \hD$.
\begin{notinextended}
This paper's extended version~\cite{Keller2015b} provides detailed plots
on the queuing delays and response times for these variants.
\end{notinextended}
\begin{inextended}
All plots in \Cref{fig:rt:scen} compares average round trip times, response times, and queuing delay (the different between response times and round trip times) as a function of resource limit $\hp$, service rate $\hMu$, demand distribution $\hD$ for the five groups of resource distributions $\hS$.
\begin{figure*}[tbp]
\centering
\begin{subfigure}{.48\textwidth}
  \centering
  \includegraphics[scale=0.9]{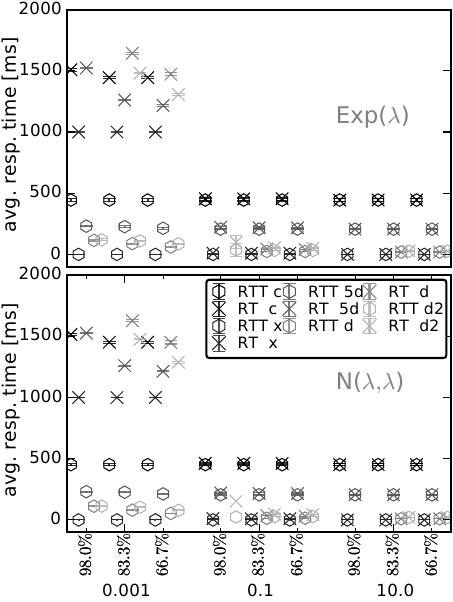}
  \caption{$\hG=\texttt{guil39}$ \label{fig:rt:scengiul39}}
\end{subfigure}
\begin{subfigure}{.48\textwidth}
  \centering
  \includegraphics[scale=0.9]{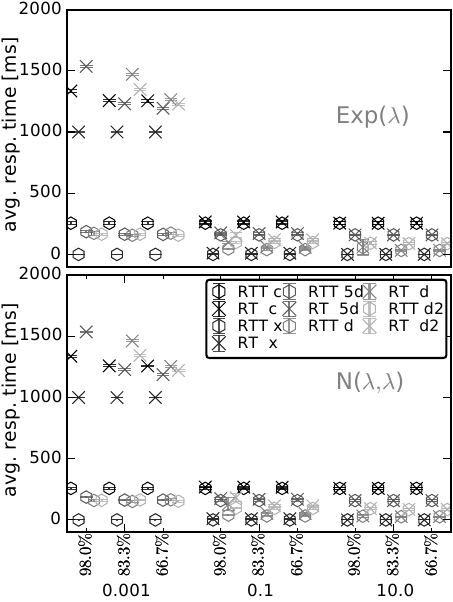}
  \caption{$\hG=\texttt{Columbus}$ \label{fig:rt:scencolumbus}}
\end{subfigure}

\begin{subfigure}{.48\textwidth}
  \centering
  \includegraphics[scale=0.9]{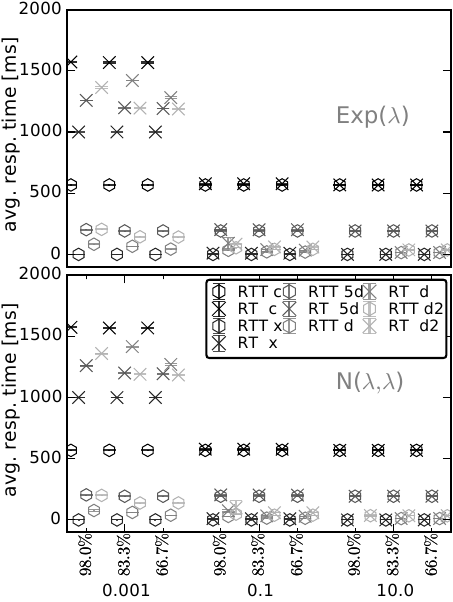}
  \caption{$\hG=\texttt{Colt}$ \label{fig:rt:scencolt}}
\end{subfigure}
\begin{subfigure}{.48\textwidth}
  \centering
  \includegraphics[scale=0.9]{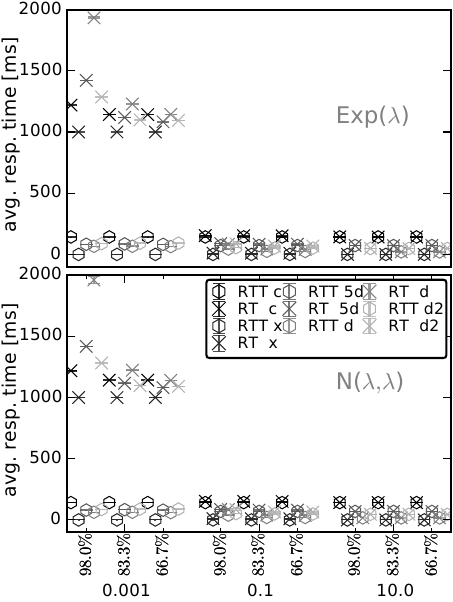}
  \caption{$\hG=\texttt{Cesnet20}$ \label{fig:rt:scencesnet20}}  
\end{subfigure}
\caption{Comparison of RT and RTT as a function of $\hMu$ and $\hp$ grouped by different~$\hS$.\label{fig:rt:scen}}
\end{figure*}

\end{inextended}

How much a 
distributed deployment reduces the response time 
compared to a single-site deployment, the quality of different
resource distributions $\hS{=}$ d, d2, 5d are compared against
$\hS{=}$c by computing quality factors. %
\begin{inextended}
\Cref{fig:rt:scenfactors} shows ECDFs for quality factors for each resource distribution.
\begin{figure}[tbp]
\centering
\includegraphics[scale=0.9]{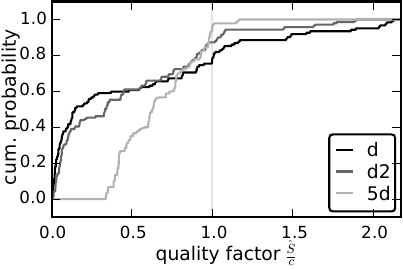}
\caption{Comparison of quality factors for three $\hS$ against $\hS{=}p$ \label{fig:rt:scenfactors}}  
\end{figure}
\end{inextended}
For these resource distributions, $78\%$, $87\%$, and $97\%$ of the
configurations yield better response times than $\hS{=}c$\ifextended{
  (the quality ratio being smaller than $1$)} and 
$61\%$, $61\%$, $35\%$ have half or shorter response times\ifextended{ (the quality ratios being smaller than $0.5$)}.
In conclusion, \emph{deploying application across multiple sites can
  (at least) halve response times}.

%% file: appendix.tex
\onecolumn
\section{Appendix}

\subsection{$\funcN_k(a)$ Derivative and Transformation }
\label{sec:apx-funcNdev}
This sections provide the transformation of first derivative $\funcN$ of $a$ with fix $k$ into a term using recursive $\funcV(a,k)$~\eqref{eq:funcv}.

\bgroup
\small
\begin{IEEEeqnarray}{ll}
    \frac{d}{da}\funcN_k(a) & =  \frac{{a}^{k}\,k\,\left( k+1\right) }{\left( k-a\right) \,\left( \left( \sum_{i=0}^{k-1}\frac{{a}^{i}}{i!}\right) \,\left( k-a\right) \,k!+{a}^{k}\,k\right) } +1
 +\frac{{a}^{k+1}\,k}{{\left( k-a\right) }^{2}\,\left( \left( \sum_{i=0}^{k-1}\frac{{a}^{i}}{i!}\right) \,\left( k-a\right) \,k!+{a}^{k}\,k\right) } 
 \notag \\
& -\frac{{a}^{k+1}\,k\,\left( \left( \sum_{i=0}^{k-1}\frac{{a}^{i-1}\,i}{i!}\right) \,\left( k-a\right) \,k!-\left( \sum_{i=0}^{k-1}\frac{{a}^{i}}{i!}\right) \,k!+{a}^{k-1}\,{k}^{2}\right) }{\left( k-a\right) \,{\left( \left( \sum_{i=0}^{k-1}\frac{{a}^{i}}{i!}\right) \,\left( k-a\right) \,k!+{a}^{k}\,k\right) }^{2}} \label{eq:derivN} \\
& = \frac{k\,\left( k+1\right)} {\left( k-a\right)\left( (k-a) \sum_{i=0}^{k-1} \frac{a^i\,k!}{i!\,a^k} + k \right) } + 1 
 + \frac{a\,k}
{{\left( k-a\right) }^{2}\left( (k-a) \sum_{i=0}^{k-1} \frac{a^i\,k!}{i!\,a^k} + k \right)} \notag \\
& - \frac{ a\,k \left( \sum_{i=0}^{k-1} \frac {a^{(i-k)}\,i\,(k-1)!}{i!} (k-a) \frac k a) + \sum_{i=0}^{k-1} \frac{a^i\,k!}{i!\,a^k} + \frac {k^2} a \right )}
{\left( k-a \right) \, \left( \sum_{i=0}^{k-1} \frac{a^i\,k!}{i!\,a^k} (k-a) + k \right)^{2}} \notag \\
& = \frac{k\,\left( k+1\right) }{{\left( k-a\right) }^{2}\,\mathrm{V}\left( a,k\right) +k\,\left( k-a\right) } +1
  + \frac{a\,k}{{\left( k-a\right) }^{3}\,\mathrm{V}\left( a,k\right) +{\left( k-a\right) }^{2}\,k}\notag \\
& - \frac{ \mathrm{V}\left( a,k-1\right) \,{k}^{2}\,\left( k-a\right) -\mathrm{V}\left( a,k\right) \,a\,k+{k}^{3} }{\left( k-a\right) \,{\left( \mathrm{V}\left( a,k\right) \,\left( k-a\right) +k\right) }^{2}} \label{eq:derivNviaV}
\end{IEEEeqnarray}
\begin{IEEEeqnarray}{ll}
V(a,k{-}1) &= \sum_{i=0}^{k-2}\frac{( k{-}1)!}{i!} {a}^{i-k+1} =
  \sum_{i=1}^{k-1}\frac{( k{-}1)!}{i!} \, i \, {a}^{i-k} 
\end{IEEEeqnarray} \notag
\egroup